\documentclass[runningheads]{llncs}
\newif\ifaplas\aplasfalse
\newif\ifdraft\draftfalse
\usepackage{amsmath,amssymb}
\usepackage{algorithm}
\usepackage{algpseudocode}
\usepackage{color}
\usepackage{centernot}
\usepackage{tikz}
\usepackage{url}
\usepackage[misc]{ifsym}
\usetikzlibrary{shapes}
\usetikzlibrary{arrows,automata}
\ifdraft
    \newcommand\nk[1]{\textcolor{blue}{[#1 -nk]}}
    \newcommand\hosoi[1]{\textcolor{red}{[#1 -hosoi]}}
\else
    \newcommand\nk[1]{}
    \newcommand\hosoi[1]{}
\fi

\newcommand{\LTS}{\mathcal{L}}
\newcommand{\States}{Q}
\newcommand{\state}{q}
\newcommand{\Actions}{\mathcal{A}}
\newcommand{\Transitions}{\longrightarrow}
\newcommand{\LTSRHS}{(\States, \Actions, \Transitions, \state_{0})}
\newcommand{\transto}[1]{\hspace{8pt}\raisebox{6pt}{\footnotesize{${#1}$}}{\hspace{-10pt}\longrightarrow}\hspace{2pt}}
\newcommand{\nottransto}[1]{\hspace{8pt}\raisebox{6pt}{\footnotesize{${#1}$}}{\hspace{-10pt}\centernot{\longrightarrow}}\hspace{2pt}}
\newcommand{\fml}{\varphi}
\newcommand{\fpo}{\alpha}
\newcommand{\hty}{\eta}
\newcommand{\obj}{{\normalfont{\texttt{o}}}}
\newcommand{\true}{{\normalfont{\texttt{true}}}}
\newcommand{\false}{{\normalfont{\texttt{false}}}}
\newcommand{\boxfml}[1]{\mbox{$\hspace{1pt}[{#1}]\hspace{1.5pt}$}}
\newcommand{\diafml}[1]{\langle{#1}\rangle}
\newcommand{\appfml}[1]{{#1}\hspace{2pt}}
\newcommand{\bndfml}[3]{{#1} #2^{{#3}}.\hspace{1pt}}
\newcommand{\lmdfml}[2]{\bndfml{\lambda}{#1}{#2}}
\newcommand{\mufml}[2]{\bndfml{\mu}{#1}{#2}}
\newcommand{\nufml}[2]{\bndfml{\nu}{#1}{#2}}
\newcommand{\HTE}{\mathcal{H}}

\newcommand{\HES}{\mathcal{E}}
\newcommand{\HESex}{\HES_{\mathtt{ex}}}
\newcommand{\fpeqn}[3]{{#1} \COL {#2} =_{{#3}}}
\newcommand{\fpeqnx}[3]{{#1} \COL {#2} \hspace{-1pt} =_{{#3}} \hspace{-2pt}}
\newcommand{\HESRHS}{(\fpeqnx{F_{1}}{\hty_{1}}{\fpo_{1}}{\fml_{1}};\cdots\hspace{-1pt};\hspace{1pt}\fpeqnx{F_{n}}{\hty_{n}}{\fpo_{n}}{\fml_{n}})}

\newcommand{\toHFL}{\mathit{toHFL}}
\newcommand{\ity}{\tau}
\newcommand{\aty}{\sigma}
\newcommand{\ITE}{\Gamma}
\newcommand{\TG}{\mathbf{TG}}
\newcommand{\verifier}{player $0$}

\newcommand{\falsifier}{player $1$}
\newcommand{\Falsifier}{Player $1$}

\newcommand{\strategy}{\varsigma}
\newcommand{\FV}{\mathit{FV}}
\newcommand{\dom}{\mathit{dom}}

\newcommand{\order}{\mathit{order}}
\newcommand{\rewritesto}{\longrightarrow}
\newcommand{\nrewritesto}{\centernot{\rewritesto}}
\newcommand{\lBrack}{\lbrack\hspace{-2pt}\lbrack\hspace{1pt}}
\newcommand{\rBrack}{\hspace{1pt}\rbrack\hspace{-2pt}\rbrack}
\newcommand{\Brack}[1]{\lBrack{#1}\rBrack}
\newcommand{\bigsqcap}{\raisebox{5pt}{\rotatebox{180}{$\bigsqcup$}}}
\newcommand{\qndprd}[2]{{#1}{#2}.\hspace{2pt}}
\newcommand{\allprd}[1]{\qndprd{\forall}{#1}}
\newcommand{\extprd}[1]{\qndprd{\exists}{#1}}
\newcommand{\Flow}{\mathit{Flow}}
\newcommand{\SG}{\mathbf{SG}}
\newcommand{\nmodels}{\nvDash}
\newcommand{\Beta}{\boldsymbol{\beta}}

\usepackage{environ}
\newcommand{\repeatlemma}[1]{%
  \begingroup
  \renewcommand{\thelemma}{\ref{#1}}%
  \expandafter\expandafter\expandafter\lemma
  \csname replemma@#1\endcsname
  \endlemma
  \endgroup
}
\NewEnviron{replemma}[1]{%
  \global\expandafter\xdef\csname replemma@#1\endcsname{%
    \unexpanded\expandafter{\BODY}%
  }%
  \expandafter\lemma\BODY\unskip\label{#1}\endlemma
}

\makeatletter
\renewcommand{\@Opargbegintheorem}[4]{%
  #4\trivlist\item[\hskip\labelsep{#3#2\@thmcounterend}]}
\makeatother

\newcommand\HFLMC{\textsc{HomuSat}}
\newcommand\LTSex{\LTS_{\mathtt{ex}}}
\newcommand\Lt[1]{\prec_{#1}}
\newcommand\Le[1]{\preceq_{#1}}
\newcommand{\Regress}{\mathit{Regress}}
\newcommand{\Forget}{\mathit{Forget}}
\newcommand\ignore[1]{}

\newcommand\IT[1]{\set{#1}}
\newcommand\set[1]{\{#1\}}
\newcommand\COL{\mathbin{:}}
\newcommand\Hole{[\,]}
\newcommand\ordinals{indices}
\newcommand\ordinal{index}
\newcommand\ITEomega[1]{\ITE_{\omega}^{(m)}}
\newcommand\ITEomegaP[1]{{\ITE'_{\omega}}^{(m)}}
\newcommand\comment[1]{}
\newcommand\horsatp{\textsc{HorSatP}}

\newcommand{\corrauthmark}{\(^{(\textrm{\Letter})}\)}

\title{A Type-Based HFL Model Checking Algorithm}

\author{Youkichi Hosoi \and Naoki Kobayashi{\corrauthmark} \and Takeshi Tsukada}
\institute{The University of Tokyo, Tokyo, Japan \\ \email{hosoi@kb.is.s.u-tokyo.ac.jp} \\ \email{koba@is.s.u-tokyo.ac.jp} \\ \email{tsukada@kb.is.s.u-tokyo.ac.jp}}
\authorrunning{Y. Hosoi, N. Kobayashi, and T. Tsukada}

\begin{document}

\maketitle

\begin{abstract}
Higher-order modal fixpoint logic (HFL) is a higher-order extension of the modal $\mu$-calculus,
and strictly more expressive than the modal $\mu$-calculus.
It has recently been shown that various program verification problems can naturally be reduced to HFL model checking:
the problem of whether a given finite state system satisfies a given HFL formula.
In this paper, we propose a novel algorithm for HFL model checking:
it is the first practical algorithm in that it runs fast for typical inputs,
despite the hyper-exponential worst-case complexity of the HFL model checking problem.
Our algorithm is based on Kobayashi et al.'s type-based characterization of HFL model checking,
and was inspired by a saturation-based algorithm for HORS model checking,
another higher-order extension of model checking.
We prove the correctness of the algorithm and report on an implementation and experimental results.
\end{abstract}

\section{Introduction}
\label{sec:intr}

Higher-order modal fixpoint logic (HFL) has been proposed by Viswanathan and Viswanathan~\cite{Viswanathan-Viswanathan:2004}.
It is a higher-order extension of the modal $\mu$-calculus and strictly more expressive than the modal \(\mu\)-calculus;
HFL can express non-regular properties of transition systems.
There have recently been growing interests in HFL model checking,
the problem of deciding whether a given finite state system satisfies a given HFL formula.
In fact, Kobayashi et al.~\cite{Kobayashi-Tsukada-Watanabe:2018,Watanabe-Tsukada-Oshikawa-Kobayashi:2019}
have shown that various verification problems for higher-order functional programs can naturally be reduced to HFL model checking problems.

Unfortunately, however, the worst-case complexity of HFL model checking is \(k\)-EXPTIME complete
(where \(k\) is a parameter called the \emph{order} of HFL formulas;
order-0 HFL corresponds to the modal \(\mu\)-calculus)~\cite{Axelson-Lange-Somla:2007},
and there has been no efficient HFL model checker.
Kobayashi et al.~\cite{Kobayashi-Lozes-Bruse:2017:POPL} have shown that there are mutual translations between HFL model checking and HORS model checking
(model checking of the trees generated by higher-order recursion schemes~\cite{Ong:2006}).
Since there are practical HORS model checkers available~\cite{Broadbent-Kobayashi:2013,Kobayashi:2009:PPDP,Kobayashi:2016,Ramsay-Neatherway-Ong:2014:POPL,Suzuki-Fujima-Kobayashi-Tsukada:2017},
one may expect to obtain an efficient HFL model checker by combining the translation from HFL to HORS model checking and a HORS model checker.
That approach does not work, however, because the translation of Kobayashi et al.~\cite{Kobayashi-Lozes-Bruse:2017:POPL} from HFL to HORS model checking
involves a complex encoding of natural numbers as higher-order functions,
which is impractical.
Considering that the other translation from HORS to HFL model checking is simpler and more natural,
we think that HFL model checking is a more primitive problem than HORS model checking.
Also in view of applications to verification of concurrent programs~\cite{Lange-Lozes-Guzman:2014,Viswanathan-Viswanathan:2004}
(in addition to the above-mentioned applications to higher-order program verification),
a direct tool support for HFL model checking is important.

In the present paper,
we propose a novel HFL model checking algorithm that is \emph{practical}
in the sense that it does not always suffer from the bottleneck of the worst-case complexity,
and runs reasonably fast for typical inputs, as confirmed by experiments.
To our knowledge, it is the first such algorithm for HFL model checking.

Our algorithm is based on Kobayashi et al.'s type-based characterization~\cite{Kobayashi-Lozes-Bruse:2017:POPL},
which reduces HFL model checking to a typability game
(which is an instance of parity games),
and was inspired by the saturation-based algorithm for HORS model checking~\cite{Suzuki-Fujima-Kobayashi-Tsukada:2017}.
The detail of the algorithm is, however, different, and its correctness is quite non-trivial.
Actually,
the correctness proof for our algorithm is simpler and more streamlined than that for their algorithm~\cite{Suzuki-Fujima-Kobayashi-Tsukada:2017}.

We have implemented a prototype HFL model checker based on the proposed algorithm.
We confirmed through experiments that the model checker works well for a number of realistic inputs obtained from program verification problems,
despite the extremely high worst-case complexity of HFL model checking.

The rest of this paper is structured as follows.
Section~\ref{sec:prlm} recalls the definition of HFL model checking,
and reviews its type-based characterization.
Section~\ref{sec:algo} formalizes our type-based HFL model checking algorithm,
and Section~\ref{sec:prfs} gives an outline of its correctness proof.
Section~\ref{sec:expr} is devoted to reporting on implementation and experimental results.
Section~\ref{sec:rlwk} discusses related work,
and Section~\ref{sec:cncl} concludes the paper.
\ifaplas
    Omitted details are found in a longer version of the paper~\cite{Hosoi-Kobayashi-Tsukada:2019:full}.
\else
    Omitted details are found in Appendix.
\fi

\section{Preliminaries}
\label{sec:prlm}

In this section,
we review the notion of HFL model checking~\cite{Viswanathan-Viswanathan:2004}
and its type-based characterization.
The latter forms the basis of our HFL model checking algorithm.

\subsection{HFL Model Checking}
\label{subsec:HFLMC}
We first review HFL model checking in this section.

A (finite) \emph{labeled transition system} (LTS) $\LTS$ is a quadruple $\LTSRHS$,
where $\States$ is a finite set of states,
$\Actions$ is a finite set of actions,
$\Transitions \, \subseteq \States \times \Actions \times \States$ is a transition relation,
and $\state_{0} \in \States$ is a designated initial state.
We use the metavariable \(a\) for actions.
We write $\state \transto{a} \state'$ when $(\state, a, \state') \in \, \Transitions$.

The \emph{higher-order modal fixpoint logic} (HFL)~\cite{Viswanathan-Viswanathan:2004}
is a higher-order extension of the modal \(\mu\)-calculus.
The sets of (simple) types and formulas are defined by the following BNF.\footnote{
    Following~\cite{Kobayashi-Lozes-Bruse:2017:POPL},
    we exclude out negations, without losing the expressive power~\cite{Lozes:2015}.
}

\[
    \begin{array}{rcl}
        \fml \mbox{ (formulas) } & ::= & \true \mid \false \mid X \mid \fml_{1} \lor \fml_{2} \mid \fml_{1} \land \fml_{2} \\[2pt]
        & & \mid \diafml{a}{\fml} \mid \boxfml{a}{\fml} \mid \mufml{X}{\hty}{\fml} \mid \nufml{X}{\hty}{\fml} \\[2pt]
        & & \mid \lmdfml{X}{\hty}{\fml} \mid \appfml{\fml_{1}}{\fml_{2}} \\[2pt]
        \hty \mbox{ (simple types) } & ::= & \obj \mid \hty_{1} \to \hty_{2}
    \end{array}
\]
The syntax of formulas on the first two lines is identical to that of the modal \(\mu\)-calculus formulas,
except that the variable \(X\) can range over \emph{higher-order} predicates,
rather than just propositions.
Intuitively,
\(\mufml{X}{\hty}{\fml}\) (\(\nufml{X}{\hty}{\fml}\), resp.)
denotes the least (greatest, resp.) predicate of type \(\hty\) such that \(X = \fml\).
Higher-order predicates can be manipulated by using \(\lambda\)-abstractions
and applications.
The type \(\obj\) denotes the type of propositions.
\ignore{
    Formulas are subject to the well-typedness condition:
    they must be simply-typed \(\lambda\)-terms,
    with the following constants:
    \[
        \begin{array}{l}
            \true, \false \COL \obj \quad
            \land, \lor \COL \obj\to\obj\to\obj\quad
            \diafml{a}{}, \boxfml{a}{} \COL \obj\to\obj\quad
            \mu^\hty, \nu^\hty \COL (\hty\to\hty)\to\hty.
        \end{array}
    \]
    Here,
    \(\mufml{X}{\hty}{\fml}\) and \(\nufml{X}{\hty}{\fml}\) are considered short hand forms of
    \(\mu^\hty(\lambda X^\hty.\fml)\) and \(\nu^\hty(\lambda X^\hty.\fml)\).
}
A \emph{type environment} $\HTE$ for simple types is a map from a finite set of variables to the set of simple types.
We often treat $\HTE$ as a set of type bindings of the form $X \COL \hty$,
and write $X \COL \hty \in \HTE$ when $\HTE(X) = \hty$.
A \emph{type judgment relation} $\HTE \vdash \fml \COL \hty$ is derived by the typing rules in Figure~\ref{fig:htyrules}.

\begin{figure}
\[
    \begin{array}{c}
        \\[-24pt]
        \begin{array}{c}
            \\ \hline
            \HTE \vdash \true \COL \obj
        \end{array} \hspace{12pt}
        \begin{array}{c}
            \\ \hline
            \HTE \vdash \false \COL \obj
        \end{array} \hspace{12pt}
        \begin{array}{c}
            X \COL \hty \in \HTE
            \\ \hline
            \HTE \vdash X \COL \hty
        \end{array} \\[12pt]

        \begin{array}{c}
            \HTE \vdash \fml_{1} \COL \obj \hspace{12pt} \HTE \vdash \fml_{2} \COL \obj
            \\ \hline
            \HTE \vdash \fml_{1} \lor \fml_{2} \COL \obj
        \end{array} \hspace{12pt}
        \begin{array}{c}
            \HTE \vdash \fml_{1} \COL \obj \hspace{12pt} \HTE \vdash \fml_{2} \COL \obj
            \\ \hline
            \HTE \vdash \fml_{1} \land \fml_{2} \COL \obj
        \end{array} \hspace{12pt}
        \begin{array}{c}
            \HTE \vdash \fml \COL \obj
            \\ \hline
            \HTE \vdash \diafml{a}{\fml} \COL \obj
        \end{array} \hspace{12pt}
        \begin{array}{c}
            \HTE \vdash \fml \COL \obj
            \\ \hline
            \HTE \vdash \boxfml{a}{\fml} \COL \obj
        \end{array} \\[12pt]

        \begin{array}{c}
            \HTE \cup \{X \COL \hty \} \vdash \fml \COL \hty \hspace{12pt} X \notin \dom(\HTE)
            \\ \hline
            \HTE \vdash \mufml{X}{\hty}{\fml} \COL \hty
        \end{array} \hspace{12pt}
        \begin{array}{c}
            \HTE \cup \{X \COL \hty \} \vdash \fml \COL \hty \hspace{12pt} X \notin \dom(\HTE)
            \\ \hline
            \HTE \vdash \nufml{X}{\hty}{\fml} \COL \hty
        \end{array} \\[12pt]

        \begin{array}{c}
            \HTE \cup \{X \COL \hty_{1} \} \vdash \fml \COL \hty_{2} \hspace{12pt} X \notin \dom(\HTE)
            \\ \hline
            \HTE \vdash \lmdfml{X}{\hty_{1}}{\fml} \COL \hty_{1} \to \hty_{2}
        \end{array} \hspace{12pt}
        \begin{array}{c}
            \HTE \vdash \fml_{1} \COL \hty_{2} \to \hty \hspace{12pt} \HTE \vdash \fml_{2} \COL \hty_{2}
            \\ \hline
            \HTE \vdash \appfml{\fml_{1}}{\fml_{2}} \COL \hty
        \end{array}
    \end{array}
\]
\caption{Typing rules for simple types}
\label{fig:htyrules}
\end{figure}

\noindent{}Note that, for each pair of a type environment $\HTE$ and an HFL formula $\fml$,
there is at most one simple type $\hty$
such that the type judgment relation $\HTE \vdash \fml \COL \hty$ is derivable.
We say an HFL formula $\fml$ has type $\hty$ under a type environment $\HTE$
if the type judgment relation $\HTE \vdash \fml \COL \hty$ is derivable.

For each simple type $\hty$, we define $\order(\hty)$ inductively by:
$\order(\obj) = 0$, $\order(\hty_{1} \to \hty_{2}) = \max(\order(\hty_{1}) + 1, \order(\hty_{2}))$.
The order of an HFL formula \(\fml\) is the highest order of the types of the variables bound by \(\mu\) or \(\nu\) in \(\fml\).
An order-0 HFL formula of type \(\obj\) can be viewed as a modal \(\mu\)-calculus formula, and vice versa.
We write $\FV(\fml)$ for the set of free variables occurring in a formula $\fml$.
An HFL formula $\fml$ is said to be \emph{closed} if $\FV(\fml) = \emptyset$,
and a closed formula is said to be \emph{well-typed} if it has some simple type under the empty type environment.

\paragraph{The semantics.}
Let $\LTS = \LTSRHS$ be an LTS.
The semantics of a well-typed HFL formula of type $\hty$
with respect to $\LTS$ is given as an element of a complete lattice
$(D_{\LTS, \hty}, \sqsubseteq_{\LTS, \hty})$
defined by induction on the structure of $\hty$.
For the base case,
$(D_{\LTS, \obj}, \sqsubseteq_{\LTS, \obj})$ is defined by
$D_{\LTS, \obj} = 2^{\States}$ and $\sqsubseteq_{\LTS, \obj} = \,\subseteq$,
that is,
$(D_{\LTS, \obj}, \sqsubseteq_{\LTS, \obj})$ is the powerset lattice of the state set $\States$.
For the step case,
$D_{\LTS, \hty_{1} \to \hty_{2}}$ is defined as the set of monotonic functions from $D_{\LTS, \hty_{1}}$ to $D_{\LTS, \hty_{2}}$,
and $\sqsubseteq_{\LTS, \hty_{1} \to \hty_{2}}$ is defined as the pointwise ordering over it.

For each type environment $\HTE$,
we define $\Brack{\HTE}_{\LTS}$ as the set of functions $\rho$ such that,
for each $X \in \dom(\HTE)$,
the image $\rho(X)$ is in the semantic domain of its type $\HTE(X)$,
that is,
$\Brack{\HTE}_{\LTS} = \{ \rho \COL \dom(\HTE) \to \mbox{$\bigcup_{\hty}$}D_{\LTS, \hty} \mid \allprd{X \COL \hty \in \HTE}{\rho(X) \in D_{\LTS, \hty}} \}$.
The interpretation of a type judgment relation $\HTE \vdash \fml \COL \hty$ is a function
$\Brack{\HTE \vdash \fml \COL \hty}_{\LTS} \COL \Brack{\HTE}_{\LTS} \to D_{\LTS, \hty}$
defined by induction on the derivation of $\HTE \vdash \fml \COL \hty$ by:
\[
    \begin{array}{l}
        \Brack{\HTE \vdash \true \COL \obj}_{\LTS}(\rho) = \States \\[2pt]
        \Brack{\HTE \vdash \false \COL \obj}_{\LTS}(\rho) = \emptyset \\[2pt]
        \Brack{\HTE \vdash X \COL \hty}_{\LTS}(\rho) = \rho(X) \\[2pt]
        \Brack{\HTE \vdash \fml_{1} \lor \fml_{2} \COL \obj}_{\LTS}(\rho) = \Brack{\HTE \vdash \fml_{1} \COL \obj}_{\LTS}(\rho) \cup \Brack{\HTE \vdash \fml_{2} \COL \obj}_{\LTS}(\rho) \\[2pt]
        \Brack{\HTE \vdash \fml_{1} \land \fml_{2} \COL \obj}_{\LTS}(\rho) = \Brack{\HTE \vdash \fml_{1} \COL \obj}_{\LTS}(\rho) \cap \Brack{\HTE \vdash \fml_{2} \COL \obj}_{\LTS}(\rho) \\[2pt]
        \Brack{\HTE \vdash \diafml{a}{\fml} \COL \obj}_{\LTS}(\rho) = \{ \state \in \States \mid \extprd{\state' \in \Brack{\HTE \vdash \fml \COL \obj}_{\LTS}(\rho)}{\state \transto{a} \state'} \} \\[2pt]
        \Brack{\HTE \vdash \boxfml{a}{\fml} \COL \obj}_{\LTS}(\rho) = \{ \state \in \States \mid \allprd{\state' \in \States}{\state \transto{a} \state' \Rightarrow \state' \in \Brack{\HTE \vdash \fml \COL \obj}_{\LTS}(\rho)} \} \\[2pt]
        \Brack{\HTE \vdash \mufml{X}{\hty}{\fml} \COL \hty}_{\LTS}(\rho) = \bigsqcap_{\LTS, \hty}\{ d\in D_{\LTS,\hty} \mid \Brack{\HTE \vdash \lmdfml{X}{\hty}{\fml} \COL \hty \to \hty}_{\LTS}(\rho)(d) \sqsubseteq_{\LTS, \hty} d \} \\[2pt]
        \Brack{\HTE \vdash \nufml{X}{\hty}{\fml} \COL \hty}_{\LTS}(\rho) = \bigsqcup_{\LTS, \hty}\{ d\in D_{\LTS,\hty} \mid d \sqsubseteq_{\LTS, \hty} \Brack{\HTE \vdash \lmdfml{X}{\hty}{\fml} \COL \hty \to \hty}_{\LTS}(\rho)(d) \} \\[2pt]
        \Brack{\HTE \vdash \lmdfml{X}{\hty_{1}}{\fml} \COL \hty_{1} \to \hty_{2}}_{\LTS}(\rho) = \lambda d\in D_{\LTS,\hty_1}.
        \Brack{\HTE \cup \{ X \COL \hty_{1} \} \vdash \fml \COL \hty_{2} }_{\LTS}(\rho[X \mapsto d]) \\[2pt]
        \Brack{\HTE \vdash \appfml{\fml_{1}}{\fml_{2}} \COL \hty}_{\LTS}(\rho) = \Brack{\HTE \vdash \fml_{1} \COL \hty_{2} \to \hty}_{\LTS}(\rho)(\Brack{\HTE \vdash \fml_{2} \COL \hty_{2}}_{\LTS}(\rho)).
    \end{array}
\]
Here, $\rho[X \mapsto d]$ denotes the function \(f\) such that
\(f(X)=d\) and \(f(Y)=\rho(Y)\) for \(Y\neq X\),
and the unary operator $\bigsqcup_{\LTS, \hty}$ ($\bigsqcap_{\LTS, \hty}$, resp.)
denotes the least upper bound (the greatest lower bound, resp.)
with respect to \(\sqsubseteq_{\LTS, \hty}\).

Finally, for each closed HFL formula $\fml$ of type $\hty$,
we define the interpretation $\Brack{\fml}_{\LTS}$ by
$\Brack{\fml}_{\LTS} = \Brack{\emptyset \vdash \fml \COL \hty}_{\LTS}(\rho_{\emptyset})$,
where $\rho_{\emptyset}$ is the empty map.
We say that a closed propositional HFL formula $\fml$ is \emph{satisfied}
by the state $\state$ when $\state \in \Brack{\fml}_{\LTS}$.

\begin{example}\label{ex:HFL}
    Let \(\fml_1\) be \(\mufml{F}{\obj \to \obj}\lmdfml{X}{\obj}X \lor \diafml{a}(\appfml{F}(\diafml{b}X))\).
    The formula \(\appfml{\fml_1}(\diafml{c}\true)\) can be expanded to:
    \[
        \begin{array}{l}
            \appfml{(\lambda X.\, X \lor \diafml{a} (\appfml{\fml_{1}}(\diafml{b}X)))} (\diafml{c}\true) \\
            \equiv \diafml{c}\true \lor \diafml{a} (\appfml{\fml_{1}}(\diafml{b}\diafml{c}\true)) \\
            \equiv \diafml{c}\true \lor \diafml{a} (\appfml{(\lambda X.\, X \lor \diafml{a} (\appfml{\fml_{1}}(\diafml{b}X)))}(\diafml{b}\diafml{c}\true)) \\
            \equiv \diafml{c}\true \lor \diafml{a} (\diafml{b}\diafml{c}\true \lor \diafml{a} (\appfml{\fml_{1}}(\diafml{b}\diafml{b}\diafml{c}\true))) \\
            \equiv \diafml{c}\true \lor \diafml{a}\diafml{b}\diafml{c}\true \lor \diafml{a}\diafml{a}\diafml{b}\diafml{b}\diafml{c}\true \lor \cdots.
        \end{array}
    \]
    Thus, the formula \(\appfml{\fml_1}(\diafml{c}\true)\) describes the property that
    there exists a transition sequence of the form \(a^nb^nc\) for some \(n\geq 0\).
    As shown by this example, HFL is strictly more expressive than the modal \(\mu\)-calculus.
    \qed
\end{example}

We write \(\LTS \models \fml\) when the initial state of \(\LTS\) satisfies the property described by \(\fml\).
The goal of HFL model checking is to decide, given \(\LTS\) and \(\fml\) as input, whether \(\LTS\models \fml\) holds.

\begin{example}
    To see how HFL model checking can be applied to program verification,
    let us consider the following OCaml-like program,
    which is a variation of the program considered in~\cite{Kobayashi-Tsukada-Watanabe:2018}.
    \begin{quote}
\begin{verbatim}
let rec f x k = if * then (close x; k())
                else (read x; read x; f x k) in
let d = open_in "foo" in f d (fun _ -> ())
\end{verbatim}
    \end{quote}
    Here, the asterisk \texttt{*} in the if-condition is a non-deterministic Boolean value.
    The program first opens the file \texttt{foo},
    and then calls the function \texttt{f} with the opened file as an argument.
    The function \texttt{f} recursively reads the given file even times and closes it upon a non-deterministic condition.

    Suppose we wish to check that the file \texttt{foo} is safely accessed as a read-only file.
    In the reduction methods by Kobayashi et al.~\cite{Kobayashi-Tsukada-Watanabe:2018},
    a program is transformed to an HFL formula that intuitively says
    ``the behavior of the program conforms to the specification described as an LTS.''
    In this case,
    the verification problem is reduced to the HFL model checking problem of deciding whether \(\LTS_{2} \models \fml_{2}\) holds,
    where
    \(\fml_{2} = \appfml{(\nufml{F}{}\lmdfml{k}{}\diafml{\mathit{close}}k \land \diafml{\mathit{read}}\diafml{\mathit{read}}(\appfml{F}k))}(\diafml{\mathit{end}}\true)\)
    and \(\LTS_{2}\) is the following LTS,
    which models the access protocol for read-only files.
    \begin{center}
    \begin{tikzpicture}[scale=0.9, every node/.style={scale=0.9}]
        \node [] (init) at (-1.2, 1) {};
        \node [state] (q0) at (0, 1) {\(\state_{0}\)};
        \node [state] (q1) at (2.0, 1) {\(\state_{1}\)};
        \node [state] (q2) at (4.0, 1) {\(\state_{2}\)};
        \draw [thick, -latex] (init) to (q0);
        \draw [thick, -latex] (q0) to [out = 90, in = 30, loop, looseness=5] node [above right, pos = 0.5] {\(\mathit{read}\)} (q0);
        \draw [thick, -latex, bend right=20] (q0) to node [below, pos = 0.5] {\(\mathit{close}\)} (q1);
        \draw [thick, -latex, bend right=20] (q1) to node [below, pos = 0.5] {\(\mathit{end}\)} (q2);
    \end{tikzpicture}
    \end{center}
    The formula \(\fml_{2}\) can be expanded to
    \(\bigwedge_{n = 0}^{\infty}\diafml{\mathit{read}}^{2n} \diafml{\mathit{close}} \diafml{\mathit{end}}\true\),
    and checking whether \(\LTS_{2} \models \fml_{2}\) holds is equivalent to checking whether
    (every prefix of) any sequence of the form \(\mathtt{read}^{2n}\cdot\mathtt{close}\cdot\mathtt{end}\)
    belongs to the prefix-closure of \(\mathtt{read}^{*}\cdot\mathtt{close}\cdot\mathtt{end}\),
    which is actually true.
    See~\cite{Kobayashi-Tsukada-Watanabe:2018}
    for systematic translations from program verification to HFL model checking.
    \qed
\end{example}

\subsection{Type-Based Characterization of HFL Model Checking}
We now review the type-based characterization of the HFL model checking problem~\cite{Kobayashi-Lozes-Bruse:2017:POPL},
which is going to be used as the basis of our algorithm given in Section~\ref{sec:algo}.

To provide the type-based characterization,
an HFL formula is represented in the form of a sequence of fixpoint equations \(\HESRHS\),
called a \emph{hierarchical equation system} (HES).
Here, for each $j \in \{1, \ldots, n\}$, $F_{j}$ is a distinct variable,
$\fpo_{j}$ is either $\mu$ or $\nu$, and $\fml_{j}$ is a fixpoint-free HFL formula that has type \(\hty_j\)
under the type environment $\{F_{1} \COL \hty_{1}, \ldots, F_{n} \COL \hty_{n}\}$.
We also require that if \(\hty_{j} = \hty_{j,1} \to \cdots \to \hty_{j,\ell} \to \obj\),
then \(\fml_j\) is of the form \(\lmdfml{X_1}{\hty_{j,1}}\cdots\lmdfml{X_\ell}{\hty_{j,\ell}}\psi_j\),
where \(\psi_j\) is a propositional formula that does not contain \(\lambda\)-abstractions.
For each HES $\HES$, we define a closed HFL formula $\toHFL(\HES)$ inductively by:
\[
    \setlength{\arraycolsep}{2pt}
    \begin{array}{l}
        \toHFL(\fpeqn{F}{\hty}{\fpo}{\fml}) = \bndfml{\fpo}{F}{\hty}{\fml} \\
        \toHFL(\HES;\hspace{1pt} \fpeqn{F}{\hty}{\fpo}{\fml}) = \toHFL(\lbrack\bndfml{\fpo}{F}{\hty}{\fml} / F\rbrack \, \HES),
    \end{array}
\]
where $[\fml / X] \, \HES$ denotes the HES obtained by
replacing all free occurrences of the variable $X$ in $\HES$ with the formula $\fml$.
Any HFL formula can be transformed to an HES, and vice versa.
For example, \(\nu X.\mu Y.(\diafml{a}X\lor \diafml{b}Y)\) can be expressed as an HES:
\(X=_\nu Y; Y=_\mu \diafml{a}X\lor \diafml{b}Y\).
We write \(\LTS\models \HES\) when \(\LTS\models \toHFL(\HES)\) holds.

Given an LTS $\LTS = \LTSRHS$,
the set of \emph{(refinement) types} for HFL formulas,
ranged over by \(\ity\), is defined by:
\[
    \ity \,\, ::= \,\, \state \mid \aty \to \ity \hspace{24pt}
    \aty \,\, ::= \,\, \IT{\ity_{1}, \ldots, \ity_{k}},
\]
where \(\state\) ranges over \(\States\).
Intuitively, \(q\) denotes the type of formulas that hold at state \(q\).
The type \(\IT{\ity_1,\ldots,\ity_k}\to\ity\) describes functions that take
a value that has type \(\ity_i\) for every \(i\in\set{1,\ldots,k}\) as input,
and return a value of type \(\ity\)
(thus, \(\set{\ity_1,\ldots,\ity_k}\) is an intersection type).
We often write $\top$ for $\emptyset$,
and \(\ity_1\land \cdots \land \ity_k\) for \(\IT{\ity_1,\ldots,\ity_k}\).
Henceforth, we just call \(\ity\) and \(\sigma\) \emph{types},
and call those ranged over by \(\hty\) \emph{simple types} or \emph{kinds}.

The \emph{refinement relations} $\ity :: \hty$ and $\aty :: \hty$,
read ``\(\ity\) and \(\aty\) are refinements of \(\hty\)'',
are inductively defined by:
\[
    \begin{array}{c}
        \state \in \States
        \\ \hline
        \state :: \obj
    \end{array} \hspace{24pt}
    \begin{array}{c}
        \allprd{\ity \in \aty}{\ity :: \hty}
        \\ \hline
        \aty :: \hty
    \end{array} \hspace{24pt}
    \begin{array}{c}
        \aty :: \hty_{1} \hspace{12pt} \ity :: \hty_{2}
        \\ \hline
        \aty \to \ity :: \hty_{1} \to \hty_{2}
    \end{array}
\]
Henceforth, we consider only those that are refinements of simple types,
excluding out ill-formed types like $\IT{\state, \state \to \state} \to \state$.

A \emph{type environment} \(\ITE\) is a finite set of type bindings of the form $X \COL \ity$,
where $X$ is a variable and $\ity$ is a type.
Note that \(\ITE\) may contain more than one type binding for the same variable.
We write $\dom(\Gamma)$ for the set $\set{ X \mid \extprd{\ity} X \COL \ity \in \ITE}$
and $\ITE(X)$ for the set $\set{ \, \ity \mid X \COL \ity \in \ITE }$.
We also write $\{ X \COL \aty \}$ for the set $\{ X \COL \ity_{1}, \ldots, X \COL \ity_{k} \}$
when $\aty = \{ \ity_{1}, \ldots, \ity_{k} \}$.
The \emph{type judgment relation} $\ITE \vdash_{\LTS} \fml \COL \ity$
for fixpoint-free formulas is defined by the typing rules in Figure~\ref{ityrules}.

\begin{figure}
\[
    \begin{array}{c}
        \begin{array}{c}
            \state \in \States
            \\ \hline
            \ITE \vdash_{\LTS} \true \COL \state
        \end{array} \hspace{8pt} (\textsc{T-True}) \hspace{16pt}
        \begin{array}{c}
            X \COL \ity \in \ITE
            \\ \hline
            \ITE \vdash_{\LTS} X \COL \ity
        \end{array} \hspace{8pt} (\textsc{T-Var}) \\[10pt]

        \begin{array}{c}
            \mbox{$\ITE \vdash_{\LTS} \fml_{i} \COL \state$ for some $i \in \{1, 2\}$}
            \\ \hline
            \ITE \vdash_{\LTS} \fml_{1} \lor \fml_{2} \COL \state
        \end{array} \hspace{8pt} (\textsc{T-Or}) \hspace{16pt}
        \begin{array}{c}
            \mbox{$\ITE \vdash_{\LTS} \fml_{i} \COL \state$ for each $i \in \{1, 2\}$}
            \\ \hline
            \ITE \vdash_{\LTS} \fml_{1} \land \fml_{2} \COL \state
        \end{array} \hspace{8pt} (\textsc{T-And}) \\[10pt]

        \begin{array}{c}
            \mbox{$\ITE \vdash_{\LTS} \fml \COL \state'$ for some $\state' \in \States$ with $\state \transto{a} \state'$}
            \\ \hline
            \ITE \vdash_{\LTS} \diafml{a}{\fml} \COL \state
        \end{array} \hspace{8pt} (\textsc{T-Some}) \\[10pt]

        \begin{array}{c}
            \mbox{$\ITE \vdash_{\LTS} \fml \COL \state'$ for each $\state' \in \States$ with $\state \transto{a} \state'$}
            \\ \hline
            \ITE \vdash_{\LTS} \boxfml{a}{\fml} \COL \state
        \end{array} \hspace{8pt} (\textsc{T-All}) \\[10pt]

        \begin{array}{c}
            \mbox{$\ITE \cup \{ X \COL \aty \} \vdash_{\LTS} \fml \COL \ity$\hspace{10pt}$X \notin \dom(\ITE)$\hspace{10pt}$\aty :: \eta$}
            \\ \hline
            \ITE \vdash_{\LTS} \lmdfml{X}{\eta}{\fml} \COL \aty \to \ity
        \end{array} \hspace{8pt} (\textsc{T-Abs}) \\[10pt]

        \begin{array}{c}
            \mbox{$\ITE \vdash_{\LTS} \fml_{1} \COL \aty \to \ity$\hspace{16pt}$\ITE \vdash_{\LTS} \fml_{2} \COL \ity'$ for each $\ity' \in \aty$}
            \\ \hline
            \ITE \vdash_{\LTS} \appfml{\fml_{1}}{\fml_{2}} \COL \ity
        \end{array} \hspace{8pt} (\textsc{T-App}) \\[10pt]

        \begin{array}{c}
            \ITE \vdash_{\LTS} \fml \COL \ity \hspace{10pt} \ity \leq_{\LTS} \ity'
            \\ \hline
            \ITE \vdash_{\LTS} \fml \COL \ity'
        \end{array} \hspace{8pt} (\textsc{T-Sub}) \hspace{16pt}
        \begin{array}{c}
            \state \in \States
            \\ \hline
            \state \leq_{\LTS} \state
        \end{array} \hspace{8pt} (\textsc{SubT-Base}) \\[10pt]

        \begin{array}{c}
            \allprd{\ity' \in \aty'}{\extprd{\ity \in \aty}{\ity \leq_{\LTS} \ity'}}
            \\ \hline
            \aty \leq_{\LTS} \aty'
        \end{array} \hspace{8pt} (\textsc{SubT-Int}) \hspace{16pt}
        \begin{array}{c}
            \aty' \leq_{\LTS} \aty \hspace{10pt} \ity \leq_{\LTS} \ity'
            \\ \hline
            \aty \to \ity \leq_{\LTS} \aty' \to \ity'
        \end{array} \hspace{8pt} (\textsc{SubT-Fun})
    \end{array}
\]
\caption{Typing rules (where \(\LTS=\LTSRHS\))}
\label{ityrules}
\end{figure}

The typability of an HES \(\HES\) is defined through the \emph{typability game} $\TG(\LTS, \HES)$,
which is an instance of parity games~\cite{Gradel-Thomas-Wilke:2002}.

\begin{definition}[Typability Game]
    Let $\LTS = \LTSRHS$ be an LTS and $\HES = \HESRHS$ be an HES with $\hty_{1} = \obj$.
    The \emph{typability game} $\TG(\LTS, \HES)$ is a quintuple $(V_{0}, V_{1}, v_{0}, E_{0} \cup E_{1}, \Omega)$, where:
\begin{itemize}
\setlength{\itemsep}{0pt}
    \item $V_{0} = \{ F_{j} \COL \ity \mid j \in \{1, \ldots, n\}, \, \ity :: \hty_{j} \}$ is the set of all type bindings.
    \item $V_{1} = \{ \, \ITE \mid \ITE\subseteq V_0\}$ is the set of all type environments.
    \item $v_{0} = F_{1} \COL \state_{0} \in V_{0}$ is the initial position.
    \item $E_{0} = \set{ (F_{j} \COL \ity, \;\ITE) \in V_{0} \times V_{1} \mid \ITE \vdash_{\LTS} \fml_{j} \COL \ity }$.
    \item $E_{1} =\set{ (\ITE, \;F_{j} \COL \ity) \in V_{1} \times V_{0} \mid F_{j} \COL \ity \in \ITE }$.
    \item $\Omega(F_{j} \COL \ity) = \Omega_{j}$ for each $F_{j} \COL \ity \in V_{0}$,
          where $\Omega_{j}$ is inductively defined by:
          $\Omega_{n} = 0$ if $\fpo_{n} = \nu$,
          $\Omega_{n} = 1$ if $\fpo_{n} = \mu$;
          and for \(i<n\), $\Omega_{i} = \Omega_{i+1}$ if $\fpo_{i} = \fpo_{i+1}$,
          and $\Omega_{i} = \Omega_{i+1} + 1$ if $\fpo_{i} \neq \fpo_{i+1}$.
          In other words, \(\Omega_{i}\ (1\leq i<n)\) is the least even (odd, resp.) number no less than \(\Omega_{i+1}\)
          if \(\fpo_i\) is \(\nu\) (\(\mu\), resp.).
    \item $\Omega(\ITE) = 0$ for all $\ITE \in V_{1}$.
\end{itemize}

    \noindent{}A typability game is a two-player game played by \verifier{} and \falsifier{}.
    The set of positions \(V_{x}\) belongs to player \(x\).
    A \emph{play} of a typability game is a sequence of positions $v_{1}v_{2}\dots$
    such that $(v_{i}, v_{i + 1}) \in E_{0} \cup E_{1}$ holds for each adjacent pair $v_{i}v_{i + 1}$.
    A maximal finite play $v_{1}v_{2}\dots v_{k}$ is won by player $x$ iff $v_{k} \in V_{1 - x}$,
    and an infinite play $v_{1}v_{2}\dots$ is won by player $x$ iff $\limsup_{i \to \infty}\Omega(v_{i}) = x \mbox{ (mod $2$)}$.
    We say a typability game is \emph{winning} if the initial position $v_{0}$ is a winning position for \verifier{},
    and call a winning strategy for her from $v_{0}$ simply a \emph{winning strategy} of the game
    (such a strategy can be given as a partial function from $V_{0}$ to $V_{1}$).
\end{definition}

Intuitively, in the position \(F_j \COL \ity\),
\verifier{} is asked to show why \(F_j\) has type \(\ity\),
by providing a type environment \(\ITE\) under which the body \(\fml_j\) of \(F_j\) has type \(\ity\).
\Falsifier{} then challenges \verifier{}'s assumption \(\ITE\),
by picking a type binding \(F'\COL\ity'\in\ITE\) and asking why \(F'\) has type \(\ity'\).
A play may continue indefinitely,
in which case \verifier{} wins if the largest priority visited infinitely often is even.

The following characterization is the basis of our algorithm.
\begin{theorem}[\cite{Kobayashi-Lozes-Bruse:2017:POPL}] \label{thm:KLB}
Let $\LTS$ be an LTS and $\HES = \HESRHS$ be an HES with $\hty_{1} = \obj$.
Then, \(\LTS\models\HES\) if and only if the typability game \(\TG(\LTS, \HES)\) is winning.
\end{theorem}

\begin{example} \label{ex:hflmc}
    Let \(\HESex{}\) be the following HES:
    \[
        S =_{\nu} \diafml{a}(\appfml{F}(\diafml{b}S));\hspace{8pt}
        F =_{\mu} \lmdfml{X}{\obj}X \lor \diafml{c}S \lor \diafml{a}(\appfml{F}(\diafml{b}X)).
    \]
    It expresses the property that
    there is an infinite sequence that can be partitioned into chunks of the form
    \(a^{k}b^{k}\) or \(a^{k}c\) (where \(k\ge 1\)),
    like \(a^3b^3a^{2}ca^{2}b^{2}a^{3}c\cdots\).

    Let \(\LTSex{}\) be an LTS shown on the left side of Figure~\ref{fig:lts}.
    It satisfies the HES \(\HESex{}\) as the sequence \(abacabac\cdots\) is enabled at the initial state \(\state_{0}\).
    The corresponding typability game \(\TG(\LTSex{}, \HESex{})\) is defined as shown (partially) on the right side of Figure~\ref{fig:lts},
    and a winning strategy (depicted by two-headed arrows) is witnessed by the type judgments
    \(\{S \COL \state_{2}, F \COL \state_{1} \to \state_{1}\} \vdash_{\LTSex{}} \fml_{S} \COL \state_{0}\),
    \(\{F \COL \top \to \state_{0}\} \vdash_{\LTSex{}} \fml_{S} \COL \state_{2}\),
    \(\emptyset \vdash_{\LTSex{}} \fml_{F} \COL \state_{1} \to \state_{1}\),
    and \(\{S \COL \state_{0}\} \vdash_{\LTSex{}} \fml_{F} \COL \top \to \state_{0}\),
    where \(\fml_{S}\) and \(\fml_{F}\) denote the right-hand side formulas of the variables \(S\) and \(F\), respectively.
    \qed
    \begin{figure}
        \begin{tikzpicture}[scale=0.86, every node/.style={scale=0.9}]
            \node [] (init) at (-1.0, 1) {};
            \node [state] (q0) at (0, 1) {\(\state_{0}\)};
            \node [state] (q1) at (1.45, 2) {\(\state_{1}\)};
            \node [state] (q2) at (2.9, 1) {\(\state_{2}\)};
            \draw [thick, -latex] (init) to (q0);
            \draw [thick, -latex] (q0) to node [above, pos = 0.35] {\(a\)} (q1);
            \draw [thick, -latex] (q1) to node [above, pos = 0.6] {\(b\)} (q2);
            \draw [thick, -latex] (q2) to node [below, pos = 0.5] {\(a\)} (q0);
            \draw [thick, -latex] (q0) to [out = 140, in = 90, loop, looseness=4] node [above, pos = 0.5] {\(c\)} (q0);

            \node [rectangle, draw] (S0) at (8.25, 0) {\(S \COL \state_{0}\)};
            \node [rectangle, draw] (S2) at (7.125, 2) {\(S \COL \state_{2}\)};
            \node [rectangle, draw] (FT0) at (5.00, 1) {\(F \COL \top \to \state_{0}\)};
            \node [rectangle, draw] (FT1) at (12.00, 2) {\(F \COL \top \to \state_{1}\)};
            \node [rectangle, draw] (F11) at (9.5, 2) {\(F \COL \state_{1} \to \state_{1}\)};

            \node [] (e) at (9.5, 3) {\(\emptyset\)};
            \node [] (s0) at (6.125, 0) {\(\{S \COL \state_{0}\}\)};
            \node [] (ft0) at (5.00, 2) {\(\{F \COL \top \to \state_{0}\}\)};
            \node [] (ft1) at (11.25, 1) {\(\{F \COL \top \to \state_{1}\}\)};
            \node [] (s2f10) at (7.125, 3) {\(\{S \COL \state_{2}, F \COL \state_{1} \to \state_{0}\}\)};
            \node [] (s2f11) at (8.25, 1) {\(\{S \COL \state_{2}, F \COL \state_{1} \to \state_{1}\}\)};

            \node [] (s0dummy) at (4.65, 3) {\(\cdots\)};
            \node [] (F10dummy) at (4.75, 0) {\(\cdots\)};
            \node [] (ft1dummy) at (3.75, 0) {\(\cdots\)};
            \node [] (FT0dummy) at (12.75, 0) {\(\cdots\)};

            \draw [thick, ->, shorten >= -2.0] (S0) to [bend right = 20] (ft1);
            \draw [thick, densely dotted, -latex, shorten <= -2.0] (ft1) to (FT1);
            \draw [thick, ->>, shorten >= -2.0] (S0) to (s2f11);
            \draw [thick, densely dotted, -latex, shorten <= -2.0] (s2f11) to (F11);
            \draw [thick, ->>, shorten >= -1.5] (F11) to (e);
            \draw [thick, densely dotted, -latex, shorten <= -2.0] (s2f11) to (S2);
            \draw [thick, ->, shorten >= -2.0] (S2.60) to (s2f10.300);
            \draw [thick, densely dotted, -latex, shorten <= -2.0] (s2f10) to (s0dummy);
            \draw [thick, densely dotted, -latex, shorten <= -2.0] (s2f10.240) to (S2.123);
            \draw [thick, ->>, shorten >= -3.0] (S2) to (ft0);
            \draw [thick, densely dotted, -latex, shorten <= -2.0] (ft0) to (FT0);
            \draw [thick, ->>, shorten >= -2.0] (FT0) to (s0);
            \draw [thick, densely dotted, -latex, shorten <= -2.0] (s0) to (S0);
            \draw [thick, ->, shorten >= -2.0] (F10dummy) to (s0);
            \draw [thick, ->] (FT0) [bend left = 10] to (ft1dummy);
            \draw [thick, ->, shorten >= -2.0] (FT0dummy) [bend left = 10] to (ft1);
        \end{tikzpicture}
        \caption{LTS \(\LTSex{}\) (on the left side) and a part of the corresponding typability game \(\TG(\LTSex{}, \HESex{})\) (on the right side)}
        \label{fig:lts}
    \end{figure}
\end{example}

\section{A Practical Algorithm for HFL Model Checking}
\label{sec:algo}

We present our algorithm for HFL model checking in this section.

Theorem~\ref{thm:KLB} immediately yields a naive model checking algorithm,
which first constructs the typability game $\TG(\LTS, \HES)$
(note that \(\TG(\LTS,\HES)\) is finite),
and solves it by using an algorithm for parity game solving.
Unfortunately, the algorithm does not work in practice,
since the size of \(\TG(\LTS,\HES)\) is too large;
it is \(k\)-fold exponential in the size of \(\LTS\) and \(\HES\),
for order-\(k\) HES.

The basic idea of our algorithm is to construct a subgame $\TG'(\LTS, \HES)$
 of $\TG(\LTS, \HES)$,
 so that $\TG'(\LTS, \HES)$ is winning if and only if the original game \(\TG(\LTS, \HES)\) is winning,
and that $\TG'(\LTS, \HES)$ is often significantly smaller than $\TG(\LTS, \HES)$.
The main question is of course how to construct such a subgame.
Our approach is to consider a series of recursion-free\footnote{
    We say an HES is recursion-free if there is no cyclic dependency on fixpoint variables,
    so that fixpoint variables can be completely
    eliminated by unfolding them;  we omit the formal definition.
} approximations \(\HES^{(0)}, \HES^{(1)},\HES^{(2)},\ldots\) of \(\HES\),
which are obtained by unfolding fixpoint variables in \(\HES\) a certain number of times,
and are free from fixpoint operators.
The key observations are:
(i) for sufficiently large \(m\) (that may depend on \(\LTS\) and \(\HES\)),
\(\LTS\models\HES^{(m)}\) if and only if \(\LTS\models\HES\),
(ii) for such \(m\),
a winning strategy for \(\TG(\LTS,\HES)\) can be constructed by using only the types used in a winning strategy for \(\TG(\LTS,\HES^{(m)})\), and
(iii) (a superset of) the types needed in a winning strategy for \(\TG(\LTS,\HES^{(m)})\) can be computed effectively
(and with reasonable efficiency for typical inputs),
based on a method similar to saturation-based algorithms for HORS model checking~\cite{Broadbent-Kobayashi:2013,Suzuki-Fujima-Kobayashi-Tsukada:2017}.
(These observations are not trivial;
they will be justified when we discuss the correctness of the algorithm in Section~\ref{sec:prfs}.)

In the rest of this section,
we first explain more details about the intuitions behind our algorithm in Section~\ref{sec:approx}.\footnote{
    Those intuitions may not be clear for non-expert readers.
    In such a case, readers may safely skip the subsection
    (except the definitions) and proceed to Section~\ref{subsec:algo}.
}
We also introduce some definitions such as \(\HES^{(m)}\) during the course of explanation.
These concepts are not directly used in the actual algorithm,
but would help readers understand intuitions behind the algorithm.
We then describe the algorithm in Section~\ref{subsec:algo}.

\subsection{The Idea of the Algorithm}
\label{sec:approx}
We first define a non-recursive HES as an approximation of $\HES$.
By the Kleene fixpoint theorem,
we can approximate $\HES$ by unfolding fixpoint variables finitely often,
and the approximation becomes exact when the depth of unfolding is sufficiently large.
Such an approximation can be naturally represented by a non-recursive HES $\HES^{(m)}$ defined as follows.
\begin{definition}[Non-Recursive HES $\HES^{(m)}$]\label{def:approx}
    Let $\HES = \HESRHS$ be an HES and $m$ be a positive integer.
    We define $\HES^{(m)} = (\fpeqnx{F_{1}^{(m)}}{\hty_{1}}{\fpo_{1}}{\fml_{1}^{(m)}}; \cdots)$
    as a non-recursive HES consisting of equations of the form
    $\fpeqnx{F_{j}^{\Beta}}{\eta_{j}}{\fpo_{j}}{\fml_{j}^{\Beta}}$.\footnote{
        Since \(\HES^{(m)}\) does not contain recursion,
        the order of equations (other than the first one) does not matter.
    }
    Here, $\Beta = (\beta_{1}, \ldots, \beta_{j})$ is a tuple of integers satisfying
    $0 \leq \beta_{k} \leq m$ for each $k \in \{1, \ldots, j\}$,
    and $\fml_{j}^{\Beta}$ is an HFL formula defined by:
    \[
        \fml_{j}^{\Beta} = \left\{
            \begin{array}{ll}
                \lmdfml{X_{1}}{\hty_{j,1}} \cdots \lmdfml{X_{\ell}}{\hty_{j,\ell}} \widehat{\fpo_{j}} & (\mbox{if \(\beta_{j} = 0\)}) \\[2pt]
                \lbrack F_{1}^{\Beta(1)}/F_{1}, \ldots, F_{n}^{\Beta(n)}/F_{n} \rbrack \,\fml_{j} & (\mbox{if \(\beta_{j} \neq 0\)}).
            \end{array}
        \right.
    \]
    Here, $\hty_{j} = \hty_{j, 1} \to \cdots \to \hty_{j, \ell} \to \obj$,
    $\widehat{\nu} = \true$, $\widehat{\mu} = \false$,
    and $\Beta(k)$ is defined by:
    \[
        \begin{array}{c}
            \Beta(k) = \left\{
                \begin{array}{ll}
                    (\beta_{1}, \ldots, \beta_{k}) & (\mbox{if \(k < j\)}) \\
                    (\beta_{1}, \ldots, \beta_{j} - 1, \underbrace{m, \ldots, m}_{\hspace{-32pt}\mbox{\footnotesize{$(k - j)$ times}}\hspace{-32pt}}) & (\mbox{if \(j \leq k\)}).\\[-17pt]
                \end{array}
            \right.\\[24pt]
        \end{array}
    \]

    \noindent{}We call the superscript $\Beta$ an \emph{\ordinal{}}.
    Intuitively, an \ordinal{} $\Beta$ indicates how many unfoldings are left to be done to obtain the formula represented by $\HES^{(m)}$.
\end{definition}

\begin{example}
    Recall the HES \(\HESex{}\) in Example~\ref{ex:hflmc}:
    \[
        S =_{\nu} \diafml{a}(\appfml{F}(\diafml{b}S));\hspace{8pt}
        F =_{\mu} \lmdfml{X}{\obj}X \lor \diafml{c}S \lor \diafml{a}(\appfml{F}(\diafml{b}X)).
    \]
    Then, a finite approximation \(\HESex^{(1)}\) is:
    \[
        \begin{array}{l}
            \hspace{18pt} S^{(1)} =_{\nu} \diafml{a}(\appfml{F^{(0, 1)}}(\diafml{b}S^{(0)})); \hspace{3pt}
            S^{(0)} =_{\nu} \true; \\
            \hspace{18pt} F^{(0, 1)} =_{\mu} \lmdfml{X}{\obj} X \lor \diafml{c}S^{(0)} \lor \diafml{a}(\appfml{{F}^{(0, 0)}}(\diafml{b}X)); \hspace{3pt}
            F^{(0, 0)} =_{\mu} \lmdfml{X}{\obj} \false; \\
            \hspace{18pt} F^{(1, 1)} =_{\mu} \lmdfml{X}{\obj} X \lor \diafml{c}S^{(1)} \lor \diafml{a}(\appfml{{F}^{(1, 0)}}(\diafml{b}X)); \hspace{3pt}
            F^{(1, 0)} =_{\mu} \lmdfml{X}{\obj} \false. \hspace{12pt} \qed
            \\[7pt]
        \end{array}
    \]
\end{example}
For \(\HES^{(m)}\), its validity can be checked by ``unfolding'' all the fixpoint variables.
The operation of unfolding is formally expressed by the following rewriting relation.

\begin{definition}[Rewriting Relation on HFL formulas]
\label{def:rewriting}
    Let $\HES = \HESRHS$ be an HES.
    The rewriting relation $\rewritesto_{\HES}$ is defined by the rule:
    \[
        C[\appfml{\appfml{\appfml{F_{j}}{\chi_{1}}}{\cdots}}{\chi_{\ell}}] \rewritesto_{\HES}
        C[[\chi_{1} / X_{1}, \ldots, \chi_{\ell} / X_{\ell}]\, \psi_{j}] \mbox{ if $\fml_{j} = \lmdfml{X_{1}}{}\cdots\lmdfml{X_{\ell}}{}\psi_{j}$}.
    \]
    Here, \(C\) ranges over the set of contexts defined by:
    \[
        C ::= \Hole \mid C\lor \chi \mid \chi\lor C \mid C\land \chi \mid \chi\land C \mid \diafml{a}C\mid \boxfml{a}C,
    \]
    and \(C[\chi]\) denotes the formula obtained from \(C\) by replacing \(\Hole\) with \(\chi\).
    We write $\rewritesto_{\HES}^{*}$ for the reflexive transitive closure of the relation $\rewritesto_{\HES}$.
\end{definition}

Note that the relation $\rewritesto_{\HES}$ preserves simple types and the semantics of formulas.
By the strong normalization property of the simply-typed \(\lambda\)-calculus,
if the HES $\HES$ does not contain recursion, it is strongly-normalizing.
Thus, for \(\HES^{(m)}\), the initial variable \(F_1^{(m)}\)
(which is assumed to have type \(\obj\))
can be rewritten to a formula \(\chi\) without any fixpoint variables,
such that an LTS \(\LTS\) satisfies \(\chi\) if and only if \(\LTS\) satisfies \(\HES^{(m)}\)
(and for sufficiently large \(m\), if and only if \(\LTS\) satisfies \(\HES\)).
Furthermore, if the initial state \(q_0\) of \(\LTS\) satisfies \(\chi\),
then from the reduction sequence:
\[
    F_1^{(m)} = \chi_0 \rewritesto_{\HES^{(m)}} \chi_1 \rewritesto_{\HES^{(m)}} \cdots \rewritesto_{\HES^{(m)}} \chi_{m'}=\chi,
\]
one can compute a series of type environments
\(\Gamma_{m'} = \emptyset, \Gamma_{m'-1},\ldots, \Gamma_1,\Gamma_0=\{F_1^{(m)}\COL q_0\}\)
such that \(\Gamma_i \vdash_{\LTS} \chi_i : \state_0\) in a backward manner,
by using the standard subject expansion property of intersection type systems
(i.e., the property that typing is preserved by backward reductions)~\cite{Barendregt-Statman-Dekkers:2013}.
These type environments provide sufficient type information,
so that a winning strategy for the typability game \(\TG(\LTS,\HES^{(m)})\)
can be expressed only by using type bindings in
\(\Gamma^{(m)} = \Gamma_{m'}\cup \cdots \cup \Gamma_0\).
If \(m\) is sufficiently large,
by using the same type bindings (but ignoring \ordinals{}),
we can also express a winning strategy for \(\TG(\LTS,\HES)\).
Thus, if we can compute (a possible overapproximation of) \(\ITE^{(m)}\) above,
we can restrict the game \(\TG(\LTS,\HES)\) to the subgame \(\TG'(\LTS,\HES)\)
consisting of only types occurring in \(\ITE^{(m)}\),
without changing the winner.

The remaining issue is how to compute an overapproximation of \(\ITE^{(m)}\).
It is unreasonable to compute it directly based on the definition above,
as the ``sufficiently large'' \(m\) is huge in general,
and the number \(m'\) of reduction steps may also be too large.
Instead, we relax the rewriting relation \(\rewritesto_\HES\) for the original HES \(\HES\) by adding the following rules:
\[
    C[F_{j}\,\chi_1\,\cdots\,\chi_\ell] \rewritesto'_\HES \left\{
        \begin{array}{ll}
            C[\true]  & \mbox{ if \(\fpo_j=\nu\)} \\
            C[\false] & \mbox{ if \(\fpo_j=\mu\)}.
        \end{array}
    \right.
\]
The resulting relation \(\rewritesto'_\HES\) simulates \(\rewritesto_{\HES^{(m)}}\) for arbitrary \(m\),
in the sense that for any reduction sequence:
\[
    F_1^{(m)} = \chi_0 \rewritesto_{\HES^{(m)}} \chi_1 \rewritesto_{\HES^{(m)}} \cdots \rewritesto_{\HES^{(m)}} \chi_{m'}=\chi,
\]
there exists a corresponding reduction sequence:
\[
    F_1 = \chi'_0 \rewritesto'_{\HES} \chi'_1 \rewritesto'_{\HES} \cdots \rewritesto'_{\HES} \chi'_{m'}=\chi,
\]
where each \(\chi'_i\) is the formula obtained by removing \ordinals{} from \(\chi_i\).
Thus, to compute \(\ITE^{(m)}\),
it suffices to compute type environments for \(\chi'_i\)'s,
based on the subject expansion property.
We can do so by using the function $\mathcal{F}_{\LTS, \HES}$ defined below,
without explicitly constructing reduction sequences.

\begin{definition}[Backward Expansion Function $\mathcal{F}_{\LTS, \HES}$]
\label{def:BEF}
    Let $\LTS$ be an LTS, and $\HES = \HESRHS$ be an HES.
    Let $\mathcal{T}_{\HES}$ denote the set of all type environments for the fixpoint variables of $\HES$,
    that is,
    $\mathcal{T}_{\HES} = \{ \, \ITE \mid \dom(\ITE) \subseteq \{ F_{1}, \ldots, F_{n} \}, \, \allprd{F_{j} \COL \ity \in \ITE}{\ity :: \hty_{j}} \}$.
    The function $\mathcal{F}_{\LTS, \HES} \COL \mathcal{T}_{\HES} \to \mathcal{T}_{\HES}$ is a monotonic function defined by:
    \[
        \setlength{\arraycolsep}{1pt}
        \begin{array}{rl}
            \mathcal{F}_{\LTS, \HES}(\ITE) = \ITE \cup \{
                & F_{j} \COL \ity \mid \ity :: \hty_{j} \\
                & \fml_{j} = \lmdfml{X_{1}}{}\cdots\lmdfml{X_{\ell}}{}\psi_{j}, \\
                & \ity = \aty_{1} \to \cdots \to \aty_{\ell} \to \state, \\
                & \extprd{\Delta}\dom(\Delta) \subseteq \FV(\psi_{j}) \cap \{ X_{1}, \ldots, X_{\ell} \}, \\
                & \ITE \cup \Delta \vdash_{\LTS} \psi_{j} \COL \state, \, \allprd{k \in \{1, \ldots, \ell\}}{\aty_{k} = \Delta(X_{k})}, \\
                & \allprd{X_{i} \in \dom(\Delta)}\extprd{\fml \in \Flow_{\HES}(X_{i})}\allprd{\ity' \in \Delta(X_{i})} \ITE \vdash_{\LTS} \fml \COL \ity'\hspace{1pt}
            \}.
        \end{array}
    \]
    Here, $\Flow_{\HES}(X_i)$ denotes the set
    \(\set{\hspace{2pt} \xi_{i} \mid F_{1} \rewritesto^*_{\HES} C[\appfml{F} \xi_{1} \cdots\hspace{1pt} \xi_{\ell}] \hspace{2pt}}\),
    where \(X_i\) is the \(i\)-th formal parameter of \(F\)
    (i.e., the equation of \(F\) is of the form
    \(F =_{\fpo} \lmdfml{X_{1}}{}\cdots\lmdfml{X_{\ell}}{}\psi\)).\footnote{
        Without loss of generality,
        we assume that \(X_1,\ldots,X_\ell\) are distinct from each other and do not occur in the other equations.
    }
\end{definition}

The following lemma justifies the definition of \(\mathcal{F}_{\LTS, \HES}\)
\ifaplas (see~\cite{Hosoi-Kobayashi-Tsukada:2019:full} for a proof).
\else (see Appendix~\ref{app:proof:algo} for a proof).
\fi
It states that the function \(\mathcal{F}_{\LTS, \HES}\) expands type environments in such a way that
we can go backwards through the rewriting relation \(\rewritesto_{\HES}\) without losing the typability.
\begin{lemma} \label{lem:inverse}
    If $F_{1} \rewritesto_{\HES}^{*} \fml \rewritesto_{\HES} \fml'$
    and $\Gamma \vdash_{\LTS} \fml' \COL \state$,
    then $\mathcal{F}_{\LTS, \HES}(\ITE) \vdash_{\LTS} \fml \COL \state$.
\end{lemma}

Let \(\ITE_0\) be the set of strongest type bindings (with respect to subtyping) for the \(\nu\)-variables of \(\HES\),
that is,
\(\ITE_{0} = \set{ F_{j} \COL \ity \mid \fpo_{j} = \nu, \, \ity :: \hty_{j}, \, \ity = \top \to \cdots \to \top \to \state }\).
The following lemma states that,
if we are allowed to use the strongest type bindings contained in \(\ITE_{0}\),
then we can also go backwards through the relaxed rewriting relation
\(\rewritesto_{\HES}'\) using the same function \(\mathcal{F}_{\LTS, \HES}\).

\begin{lemma} \label{lem:inverse2}
    If $F_{1} {\rewritesto'_{\HES}}^{*} \fml \rewritesto'_{\HES} \fml'$,
    $\Gamma \vdash_{\LTS} \fml' \COL \state$,
    and \(\ITE\hspace{-1pt}\supseteq\hspace{-1pt}\ITE_0\),
    then $\mathcal{F}_{\LTS, \HES}(\ITE) \vdash_{\LTS} \fml \COL \state$.
\end{lemma}

Let us write \((\mathcal{F}_{\LTS,\HES})^\omega(\Gamma_0)\) for \(\bigcup_{i\in\omega}(\mathcal{F}_{\LTS, \HES})^i(\Gamma_0)\).
As an immediate corollary of Lemma~\ref{lem:inverse2}, we have:
if \(F_1 = \chi'_0 \rewritesto'_{\HES} \chi'_1 \rewritesto'_{\HES} \cdots \rewritesto'_{\HES} \chi'_{m'} = \chi\)
and \(\emptyset \vdash_{\LTS} \chi \COL \state_0\),
then \((\mathcal{F}_{\LTS,\HES})^\omega(\Gamma_0) \vdash_{\LTS} \chi'_i \COL \state_0\) for every \(i\).
Thus, \((\mathcal{F}_{\LTS, \HES})^\omega(\Gamma_0)\) serves as an overapproximation of \(\ITE^{(m)}\) mentioned above.

\subsection{The Algorithm}
\label{subsec:algo}

\begin{algorithm}
    \caption{The proposed HFL model checking algorithm}
    \label{lst:pseudocode}
    \begin{algorithmic}
        \State $\ITE := \ITE_{0}$
        \While{$\ITE \neq \mathcal{F}_{\LTS, \HES}'(\ITE)$}
            \State $\Gamma := \mathcal{F}_{\LTS, \HES}'(\ITE)$
        \EndWhile
        \State \Return whether the subgame $\SG(\LTS, \HES, \Gamma)$ is winning
    \end{algorithmic}
\end{algorithm}

Based on the intuitions explained in Section~\ref{sec:approx},
we propose the algorithm shown in Algorithm~\ref{lst:pseudocode}.

In the algorithm,
the function $\mathcal{F}_{\LTS, \HES}'$ is an overapproximation of the function $\mathcal{F}_{\LTS, \HES}$,
obtained by replacing $\Flow_{\HES}$ in the definition of $\mathcal{F}_{\LTS, \HES}$
with an overapproximation $\Flow_{\HES}'$ satisfying $\allprd{X}{\Flow_{\HES}(X) \subseteq \Flow_{\HES}'(X)}$.
This is because it is in general too costly to compute the exact flow set $\Flow_{\HES}(X)$.
The overapproximation $\Flow_{\HES}'$ can typically be computed by flow analysis algorithms for functional programs,
such as 0-CFA~\cite{Shivers:1991}.
The first four lines compute \(\bigcup_{i\in\omega}(\mathcal{F}'_{\LTS, \HES})^i(\Gamma_0)\),
which is an overapproximation of \(\bigcup_{i\in\omega}(\mathcal{F}_{\LTS, \HES})^i(\Gamma_0)\) discussed in the previous subsection.

\(\SG(\LTS,\HES,\ITE)\) on the last line denotes the subgame of \(\TG(\LTS,\HES)\),
obtained by restricting the game arena.
It is defined as follows.
\begin{definition}[Subgame]
\label{def:subgame}
    Let $\LTS = \LTSRHS$ be an LTS,
    $\HES = \HESRHS$ be an HES with $\hty_{1} = \obj$,
    and $\ITE \in \mathcal{T}_{\HES}$ be a type environment for $\HES$.
    The \emph{subgame} $\SG(\LTS, \HES, \ITE)$ is a parity game defined the same as $\TG(\LTS, \HES)$
    except that the set of positions is restricted to the subsets of $\ITE$.
    That is,
    for \(\TG(\LTS,\HES)=(V_0',V_1',v_{0}',E'_0\cup E'_1,\Omega')\),
    $\SG(\LTS, \HES, \ITE)$ is the parity game $(V_{0}, V_{1}, v_{0}, E_{0} \cup E_{1}, \Omega)$,
    where:
    \begin{itemize}
        \item $V_{0} = \ITE \cup \{v_{0}'\}$,
              $V_{1} = \{ \, \ITE' \mid \ITE' \subseteq \ITE \}$, \\
              $v_{0} = v_{0}'$,
              $E_{0} = E_0'\cap (V_{0} \times V_{1})$,
              $E_{1} = E_1'\cap (V_{1} \times V_{0})$.
        \item \(\Omega\) is the restriction of \(\Omega'\) to \(V_0\cup V_1\).
    \end{itemize}
\end{definition}

The following theorem claims the correctness of the algorithm.
\begin{theorem}[Correctness] \label{thm:main}
    Let $\LTS$ be an LTS and $\HES = \HESRHS$ be an HES with $\hty_{1} = \obj$.
    Algorithm~\ref{lst:pseudocode} terminates.
    Furthermore, it returns ``yes'' if and only if \(\LTS\models \HES\).
\end{theorem}

\begin{example}
    Recall the HES \(\HESex{}\) and the LTS \(\LTSex{}\) in Example~\ref{ex:hflmc}.
    The fixpoint computation from the initial type environment
    \(\ITE_{0} = \{S \COL \state_{0}, \, S \COL \state_{1}, \, S \COL \state_{2} \}\)
    by the function \(\mathcal{F}_{\LTSex{}, \HESex{}}\)
    (with a few simple optimizations)\footnote{
        Using subtyping relations,
        we can refrain from unnecessary type derivations like
        \(\{S \COL \state_{0}, X \COL \state_{1}\} \vdash_{\LTSex{}} \psi_{F} \COL \state_{0}\)
        in this example.
        \ifaplas
            See~\cite{Hosoi-Kobayashi-Tsukada:2019:full} for more details.
        \else
            See Appendix~\ref{app:opt} for more details.
        \fi
    } proceeds as shown in Table~\ref{tab:expansion}.
    Note that the flow set \(\Flow_{\HESex{}}(X)\)
    for the formal parameter \(X\) of the fixpoint variable \(F\)
    is calculated as \(\{\diafml{b}S, \, \diafml{b}\diafml{b}S, \ldots\}\),
    and thus the only candidates for the type environment \(\Delta\)
    in the algorithm are \(\Delta = \emptyset\) and \(\Delta = \{X \COL \state_{1}\}\).
    The expansion reaches the fixpoint after two iterations,\footnote{
        Actually, our prototype model checker reported in Section~\ref{sec:expr}
        does not derive the type binding \(F \COL \top \to \state_{2}\),
        and thus the computation terminates in one iteration.
        This is because it uses information of types in flow analysis,
        which reveals that it does not affect the result whether formulas of the form \(\appfml{F} \fml\) have type \(\state_{2}\).
        \ifaplas
            See~\cite{Hosoi-Kobayashi-Tsukada:2019:full} for details.
        \else
            See Appendix~\ref{app:opt}.
        \fi
    } and the algorithm returns ``yes'' as the resulting type environment contains
    sufficient type bindings to construct the winning strategy depicted in Figure~\ref{fig:lts} of Example~\ref{ex:hflmc}.
    \qed
    \begin{table}
        \caption{Fixpoint computation by the function \(\mathcal{F}_{\LTSex{}, \HESex{}}\)}
        \begin{tabular}{|c|c|c|}
            \hline
            Iteration number $k$ &
            Type environment $\ITE_{k}$ &
            Newly derivable type judgments
            \\ \hline
            0 &
            \(
                \{S \COL \state_{0}, \, S \COL \state_{1}, \, S \COL \state_{2}\}
            \) &
            \( \setlength{\arraycolsep}{1pt}\begin{array}{ccc}
                \{X \COL \state_{1}\} & \vdash_{\LTSex{}} & \psi_{F} \COL \state_{1} \\
                \{S \COL \state_{0}\} & \vdash_{\LTSex{}} & \psi_{F} \COL \state_{0}
                \end{array}
            \)
            \\ \hline
            1 &
            \(
                \ITE_{0} \cup \{F \COL \state_{1} \to \state_{1}, \, F \COL \top \to \state_{0}\}
            \) &
            \( \begin{array}{c}
                    \{F \COL \top \to \state_{0}\} \vdash_{\LTSex{}} \psi_{F} \COL \state_{2} \\
                \end{array}
            \)
            \\ \hline
            2 &
            \(
                \ITE_{1} \cup \{F \COL \top \to \state_{2}\}
            \) & - \\ \hline
        \end{tabular}
        \label{tab:expansion}
    \end{table}
\end{example}

\section{Correctness of the Algorithm}
\label{sec:prfs}

We sketch a proof of Theorem~\ref{thm:main} in this section.
A more detailed proof is found in
\ifaplas
    \cite{Hosoi-Kobayashi-Tsukada:2019:full}.
\else
    Appendix~\ref{app:proof:prfs}.
\fi
We discuss soundness and completeness (Theorems~\ref{thm:sound} and~\ref{thm:complete} below) separately,
from which Theorem~\ref{thm:main} follows.

\subsection{Soundness of the Algorithm}
\label{subsec:soundness}

The soundness of the algorithm follows immediately from the fact that the replacement of \(\TG(\LTS,\HES)\)
with the subgame $\SG(\LTS, \HES, \ITE)$ restricts only the moves of \verifier{},
so that the resulting game is harder for her to win.

\begin{theorem}[Soundness] \label{thm:sound}
    Let $\LTS$ be an LTS and $\HES = \HESRHS$ be an HES with $\hty_{1} = \obj$.
    If $\LTS \nmodels \HES$, then the algorithm returns ``no'', that is,
    the subgame $\SG(\LTS, \HES, (\mathcal{F}_{\LTS, \HES}')^{\omega}(\ITE_{0}))$ is not winning.
\end{theorem}

\begin{proof}
    We show the contraposition.
    Suppose that $\SG(\LTS, \HES, (\mathcal{F}_{\LTS, \HES}')^{\omega}(\ITE_{0}))$ is a winning game.
    Then, there exists a winning strategy $\strategy$ of \verifier{} for the node $F_{1} \COL q_{0}$ in that game.
    This strategy $\strategy$ also gives a winning strategy of \verifier{} for the node $F_{1} \COL \state_{0}$
    in the original typability game $\TG(\LTS, \HES)$;
    note that for each position \(\ITE\in V_1\) of \(\SG(\LTS, \HES, (\mathcal{F}_{\LTS, \HES}')^{\omega}(\ITE_{0}))\),
    the set of possible moves of \falsifier{} in \(\TG(\LTS, \HES)\)
    is the same as that in \(\SG(\LTS, \HES, (\mathcal{F}_{\LTS, \HES}')^{\omega}(\ITE_{0}))\).
    Therefore, $\LTS \models \HES$ follows from Theorem~\ref{thm:KLB}.
    \qed
\end{proof}

\subsection{Completeness of the Algorithm}
\label{subsec:completeness}

The completeness of the algorithm is stated as Theorem~\ref{thm:complete} below.

\begin{theorem}[Completeness] \label{thm:complete}
    Let $\LTS$ be an LTS and $\HES = \HESRHS$ be an HES with $\hty_{1} = \obj$.
    If $\LTS \models \HES$, then the algorithm returns ``yes'', that is,
    the subgame $\SG(\LTS, \HES, (\mathcal{F}_{\LTS, \HES}')^{\omega}(\ITE_{0}))$ is winning.
\end{theorem}

The proof follows the intuitions provided in Section~\ref{sec:approx}.
Given a type environment \(\ITE\) for \(\HES^{(m)}\),
we write \(\Forget(\ITE)\) for the type environment obtained by removing all the \ordinals{} from \(\ITE\).
Theorem~\ref{thm:complete} follows immediately from Lemmas~\ref{lem:approx-is-complete} and~\ref{lem:overapproximation-of-Em} below.
\begin{lemma}
\label{lem:approx-is-complete}
    If the typability game \(\TG(\LTS,\HES)\) is winning,
    then for sufficiently large \(m\),
    the subgame $\SG(\LTS, \HES, \Forget((\mathcal{F}_{\LTS, \HES^{(m)}})^{\omega}(\emptyset)))$ is also winning.
\end{lemma}

\begin{lemma}
\label{lem:overapproximation-of-Em}
    \(\Forget((\mathcal{F}_{\LTS, \HES^{(m)}})^{\omega}(\emptyset))) \subseteq (\mathcal{F}_{\LTS, \HES})^{\omega}(\ITE_0)\).
\end{lemma}
Note that Lemma~\ref{lem:overapproximation-of-Em} implies that
the game $\SG(\LTS, \HES, (\mathcal{F}_{\LTS, \HES}')^{\omega}(\ITE_{0}))$ is more advantageous for \verifier{}
than the game $\SG(\LTS, \HES, \Forget((\mathcal{F}_{\LTS, \HES^{(m)}})^{\omega}(\emptyset)))$,
which is winning when \(\LTS \models \HES\) by Lemma~\ref{lem:approx-is-complete}.

Lemma~\ref{lem:overapproximation-of-Em} should be fairly obvious,
based on the intuitions given in Section~\ref{sec:approx}.
Technically, it suffices to show that
\(F_j^{\Beta}\COL\ity \in (\mathcal{F}_{\LTS, \HES^{(m)}})^{i}(\emptyset)\) implies
\(F_j \COL \ity \in (\mathcal{F}_{\LTS, \HES})^{i}(\ITE_0)\) by induction on \(i\),
with case analysis on \(\beta_j\).
If \(\beta_j = 0\), then \(\fpo_j = \nu\) and the body of \(F_j^{\Beta}\) is \(\lmdfml{\widetilde{X}}{}\true\).
Thus, \(F_j \COL \ity = F_j \COL \top \to \cdots \to \top \to q \in \ITE_0 \subseteq (\mathcal{F}_{\LTS, \HES})^{i}(\ITE_0)\).
If \(\beta_j>0\),
then the body of \(F_j^{\Beta}\) is the same as that of \(F_j\) except \ordinals{}.
Thus, \(F_j \COL \ity \in (\mathcal{F}_{\LTS, \HES})^{i}(\ITE_0)\)
follows from the induction hypothesis and the definition of the function \(\mathcal{F}_{\LTS, \HES}\).
\ifaplas
    See \cite{Hosoi-Kobayashi-Tsukada:2019:full} for details.
\else
    See Appendix~\ref{app:proof:prfs} for details.
\fi

To prove Lemma~\ref{lem:approx-is-complete},
we define another function \(\Regress\) on type environments.
Let \(\Le{k}\) be the lexicographic ordering on the first \(k\) elements of tuples of integers,
and \(\Lt{k}\) be its strict version.
We write \(\Beta_1=_k\Beta_2\) if \(\Beta_1\Le{k}\Beta_2\) and \(\Beta_2\Le{k}\Beta_1\).
For example, $(1, 2) =_{0} (1, 2, 3)$, $(1, 2) =_{1} (1, 2, 3)$,
$(1, 2) =_{2} (1, 2, 3)$, and $(1, 2) \Lt{3} (1, 2, 3)$.
Note that \ordinals{} \(\Beta\) combined with the order \(\Le{k}\)
can be used to witness a winning strategy of a parity game through a proper assignment of them
(a \emph{parity progress measure}~\cite{Jurdzinski:2000}) to positions of the game.
The function \(\Regress\) is defined as follows.
\begin{definition}[Function $\Regress$]
    First, we define $\ITE_{\mu}^{\Beta}$ and $\ITE_{\nu}^{\Beta}$
    for a type environment $\ITE$ for $\HES^{(m)}$ and an \ordinal{} $\Beta$ of length $j$ by:
    \[
        \begin{array}{l}
            \ITE_{\mu}^{\Beta} = \{ F_{j'} \COL \ity \mid \extprd{F_{j'}^{\Beta'} \COL \ity \in \ITE} \Beta' \prec_{j} \Beta \} \\
            \ITE_{\nu}^{\Beta} = \{ F_{j'} \COL \ity \mid \extprd{F_{j'}^{\Beta'} \COL \ity \in \ITE} \Beta' \preceq_{j - 1} \Beta \},
        \end{array}
    \]
    that is, $\ITE_{\fpo}^{\Beta}$ is a type environment for the original HES $\HES$
    consisting of all type bindings in $\ITE$ with \ordinals{} ``smaller'' than $\Beta$
    (the meaning of ``smaller'' depends on the fixpoint operator $\fpo$).
    Using this $\ITE_{\fpo}^{\Beta}$,
    we define $\Regress$ as a monotonic function on type environments for $\HES^{(m)}$ by:
    \[
        \Regress(\ITE) = \{ F_{j}^{\Beta} \COL \ity \in \ITE \mid \ITE_{\fpo_{j}}^{\Beta} \vdash_{\LTS} \fml_{j} \COL \ity \},
    \]
    that is, $\Regress(\ITE)$ consists of all $F_{j}^{\Beta} \COL \ity \in \ITE$
    such that the right-hand side formula $\fml_{j}$ of $F_{j}$ in the original HES $\HES$
    has type $\ity$ under the type environment $\ITE_{\fpo_{j}}^{\Beta}$.
\end{definition}

Note that \(\Regress\) is a monotonic function on a finite domain.
We write \(\Regress^\omega(\ITE)\) for \(\bigcap_{i\in\omega}\Regress^i(\ITE)\),
which is the greatest \(\ITE'\) such that \(\ITE'\subseteq \ITE\) and \(\Regress(\ITE')=\ITE'\).
Lemma~\ref{lem:approx-is-complete} follows immediately from the following two lemmas
(Lemmas~\ref{lem:regress} and~\ref{lem:regress2}).

\begin{lemma}
\label{lem:regress}
    If \(\ITE\subseteq (\mathcal{F}_{\LTS, \HES^{(m)}})^{\omega}(\emptyset)\) is a fixpoint of \(\Regress\),
    then \(\Forget(\ITE)\) is a subset of the winning region of \verifier{} for
    $\SG(\LTS, \HES, \Forget((\mathcal{F}_{\LTS, \HES^{(m)}})^{\omega}(\emptyset)))$.
\end{lemma}
This is intuitively because,
for each \(F_{j} \COL \ity \in \Forget(\ITE)\),
we can find \(F_{j}^{\Beta} \COL \ity \in \ITE\)
such that choosing \(\Gamma_{\fpo_{j}}^{\Beta}\) at \(F_{j} \COL \ity\) gives a winning strategy
\ifaplas
    for \verifier{}.
\else
    for \verifier{}; see Appendix~\ref{app:proof:prfs} for details.
\fi
Now it remains to show:
\begin{lemma}
\label{lem:regress2}
    If the typability game \(\TG(\LTS,\HES)\) is winning,
    then for sufficiently large \(m\),
    \(F_1 \COL q_0 \in \Forget(\Regress^\omega((\mathcal{F}_{\LTS, \HES^{(m)}})^{\omega}(\emptyset)))\).
\end{lemma}
We prepare a few further definitions and
\ifaplas
    lemmas.
\else
    lemmas (see Appendix~\ref{app:proof:prfs} for proofs).
\fi
Let $\ITEomega{m}$ be \((\mathcal{F}_{\LTS, \HES^{(m)}})^{\omega}(\emptyset)\)
and $\ITEomegaP{m}$ be \(\Regress^{\omega}(\ITEomega{m})\).
For each $k = 1, 2, \ldots$,
we define $D_{k}$ as the set of type bindings removed by the $k$-th application of $\Regress$ to $\ITEomega{m}$,
that is, $D_{k} = \Regress^{k - 1}(\ITEomega{m}) \backslash \Regress^{k}(\ITEomega{m})$.

\begin{replemma}{lem:seqbase}
    If $F_{j}^{\Beta} \COL \ity \in D_{k}$ and $\beta_{j} = 0$,
    then $\fpo_{j} = \nu$.
\end{replemma}

\begin{replemma}{lem:seqstep}
    If $F_{j}^{\Beta} \COL \ity \in D_{k}$ and $\beta_{j} \neq 0$,
    then there exists $k'$ satisfying $1 \leq k' < k$
    such that $F_{j'}^{\Beta(j')} \COL \ity' \in D_{k'}$ holds for some $j'$ and $\ity'$.
\end{replemma}

We are now ready to prove Lemma~\ref{lem:regress2}.
\begin{proof}[Proof of Lemma~\ref{lem:regress2}]
    We show the lemma by contradiction.
    Since the typability game \(\TG(\LTS,\HES)\) is winning,
    we have $F_{1}^{(m)} \COL q_{0} \in \ITEomega{m}$ for sufficiently large \(m\)
    (this is intuitively because \(\TG(\LTS,\HES^{(m)})\) is also winning;
    \ifaplas
        see \cite{Hosoi-Kobayashi-Tsukada:2019:full}
    \else
        see Appendix~\ref{app:proof:prfs}, Lemma~\ref{lem:inverse3}
    \fi
    for a formal proof).
    Suppose it were the case that $F_{1}^{(m)} \COL q_{0} \notin \ITEomegaP{m}$.
    Then there must be a positive integer $k$ such that $F_{1}^{(m)} \COL q_{0} \in D_{k}$.
    Therefore, by Lemma~\ref{lem:seqstep},
    there exists a sequence of type bindings
    $F_{1}^{(m)} \COL q_{0} = F_{j_{0}}^{\Beta_{0}} \COL \ity_{0}, \, F_{j_{1}}^{\Beta_{1}} \COL \ity_{1}, \, \ldots, \, F_{j_{\ell}}^{\Beta_{\ell}} \COL \ity_{\ell}$
    such that
    (i) $\Beta_{\ell}$ ends with 0, and
    (ii) $\Beta_{i} = \Beta_{i - 1}(j_{i})$ and $F_{j_{i}}^{\Beta_{i}} \COL \ity_{i} \in D_{k_{i}}$ hold for each $i \in \{1, \ldots, \ell\}$,
    where $k = k_{0} > k_{1} > \cdots > k_{\ell}$.
    Moreover, we have $\fpo_{j_{\ell}} = \nu$ by Lemma~\ref{lem:seqbase}.
    Let $\Beta_{\ell} = (\beta_{1}, \ldots, \beta_{j_{\ell} - 1}, 0)$.
    Then, $(\beta_{1}, \ldots, \beta_{j_{\ell} - 1}, m)$,
    $(\beta_{1}, \ldots, \beta_{j_{\ell} - 1}, m - 1)$, $\ldots$,
    and $(\beta_{1}, \ldots, \beta_{j_{\ell} - 1}, 1)$
    must exist in the sequence $\Beta_{0}, \, \Beta_{1}, \, \ldots, \, \Beta_{\ell - 1}$
    \ifaplas
        in this order (see \cite{Hosoi-Kobayashi-Tsukada:2019:full} for a proof).
    \else
        in this order (see Appendix~\ref{app:proof:prfs}, Lemma~\ref{lem:beta}).
    \fi

    For each $i \in \{0, \ldots, m\}$,
    let $\ell_{i}$ be the integer with $\Beta_{\ell_{i}} = (\beta_{1}, \ldots, \beta_{j_{\ell} - 1}, i)$.
    Since the number of intersection types $\ity'$ satisfying $\ity' :: \hty_{j_{\ell}}$ is finite,
    there must exist duplicate types in the sequence
    $\ity_{\ell_{0}}, \ity_{\ell_{1}}, \ldots, \ity_{\ell_{m}}$ for sufficiently large $m$.
    Let $\ity_{\ell_{a}}$ and $\ity_{\ell_{b}}$ be such a pair with $\ell_{a} < \ell_{b}$.
    Then, we have $F_{j_{\ell}}^{\Beta_{\ell_{a}}} \COL \ity_{\ell_{a}} \in D_{k_{\ell_{a}}}$
    and $F_{j_{\ell}}^{\Beta_{\ell_{b}}} \COL \ity_{\ell_{b}} \in D_{k_{\ell_{b}}}$.
    However, since $\fpo_{j_{\ell}} = \nu$ and $\Beta_{\ell_{a}} =_{j_{\ell} - 1} \Beta_{\ell_{b}}$,
    we have $\ITE_{\fpo_{j_{\ell}}}^{\Beta_{\ell_{a}}} = \ITE_{\fpo_{j_{\ell}}}^{\Beta_{\ell_{b}}}$ for any $\ITE$.
    Therefore, by the definition of the function $\Regress$ and the assumption $\ity_{\ell_{a}} = \ity_{\ell_{b}}$,
    the type bindings $F_{j_{\ell}}^{\Beta_{\ell_{a}}} \COL \ity_{\ell_{a}}$
    and $F_{j_{\ell}}^{\Beta_{\ell_{b}}} \COL \ity_{\ell_{b}}$
    must be removed by $\Regress$ at the same time.
    This contradicts the assumption $\ell_{a} < \ell_{b}$.
    Therefore, $F_{1}^{(m)} \COL q_{0} \in \ITEomegaP{m}$ holds for sufficiently large $m$.
    \qed
\end{proof}

\section{Implementation and Experiments}
\label{sec:expr}

We have implemented a prototype HFL model checker \HFLMC{}\footnote{
    The source code and the benchmark problems used in the experiments are available at
    \url{https://github.com/hopv/homusat}.
} based on the algorithm discussed in Section~\ref{sec:algo}.
As mentioned in footnotes in Section~\ref{sec:algo},
some optimization techniques are used to improve the performance of \HFLMC{}.
\ifaplas
    See \cite{Hosoi-Kobayashi-Tsukada:2019:full} for an explanation on these optimizations.
\else
    See Appendix~\ref{app:opt} for a brief explanation on several of these optimizations.
\fi

We have carried out experiments to evaluate the efficiency of \HFLMC{}.
As benchmark problems,
we used HORS model checking problems used as benchmarks for HORS model checkers
\textsc{TravMC2}~\cite{Neatherway-Ong:2014},
\textsc{HorSatP}~\cite{Suzuki-Fujima-Kobayashi-Tsukada:2017},
and \textsc{HorSat2}~\cite{Kobayashi:2016}.
These benchmarks include many typical instances of higher-order model checking,
such as the ones obtained from program verification problems.
They were converted to HFL model checking problems via the translation by Kobayashi et al.~\cite{Kobayashi-Lozes-Bruse:2017:POPL}.
The resulting set of benchmarks consists of 136 problems of orders up to 8.
Whereas the LTS size is moderate (around 10) for most of the instances,
there are several instances with large state sets (including those with $|Q| > 100$).
HES sizes vary from less than 100 to around 10,000;
note that in applications to higher-order program verification~\cite{Kobayashi-Tsukada-Watanabe:2018,Watanabe-Tsukada-Oshikawa-Kobayashi:2019},
the size of an HES corresponds to the size of a program to be verified.
As to the number of alternations between \(\mu\) and \(\nu\) within the HES,
over half of the instances (83 out of 136) have no alternation
(that is, they are \(\mu\)-only or \(\nu\)-only),
but there are a certain number of instances that have one or more alternations,
up to a maximum of 4.
\ifaplas\else
    See Appendix~\ref{sec:bench} for the distributions of the orders,
    the LTS sizes,
    and the numbers of alternations between \(\mu\) and \(\nu\) in the benchmark set.
\fi
The experiments were conducted on a machine with 2.3 GHz Intel Core i5 processor and 8 GB memory.
As a reference, we have compared the result with \textsc{HorSatP},
one of the state-of-the-art HORS model checkers,\footnote{
    For the restricted class of properties expressed by trivial automata,
    \textsc{HorSat2} is the state-of-the-art.
} run for the original problems.

The results are shown in Figures~\ref{fig:results} and~\ref{fig:ht}.
Figure~\ref{fig:results} compares the running times of \HFLMC{} with those of \horsatp{}.
As the figure shows, \HFLMC{} often outperforms \textsc{HorSatP}.
Although it is not that this result indicates the proposed algorithm is superior as a higher-order model checking algorithm to \textsc{HorSatP}
(the two model checkers differ in the degree of optimization),\footnote{
    Actually,
    as the two algorithms are both based on type-based saturation algorithm~\cite{Broadbent-Kobayashi:2013},
    various type-oriented optimization techniques used in \HFLMC{} can also be adapted to \textsc{HorSatP} and are expected to improve its performance.
} the fact that \HFLMC{} works fast for various problems obtained via the mechanical conversion from HORS to HFL,
which increases the size of inputs and thus makes them harder to solve,
is promising.
Evaluation of the efficiency of the proposed algorithm against a set of problems obtained directly as HFL model checking problems is left for future work.

Figure~\ref{fig:ht} shows the distribution of the running times of \HFLMC{} with respect to the input HES size.
As the figure shows, despite the \(k\)-EXPTIME worst-case complexity,
the actual running times do not grow so rapidly.
This is partially explained by the fact that the time complexity of HFL model checking is
fixed-parameter polynomial in the size of HES~\cite{Kobayashi-Lozes-Bruse:2017:POPL}.

\begin{figure}[t]
    \input{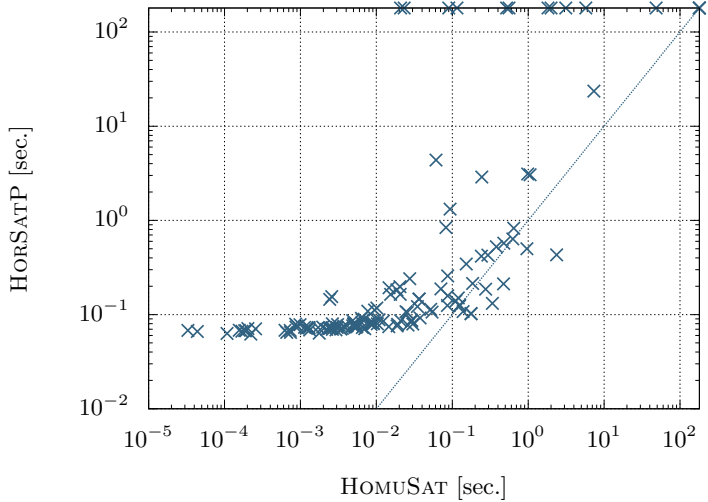}
    \caption{The experimental results: comparison with \horsatp{} ($\mbox{timeout} = \mbox{180 sec.}$)}
    \label{fig:results}
\end{figure}

\begin{figure}
    \input{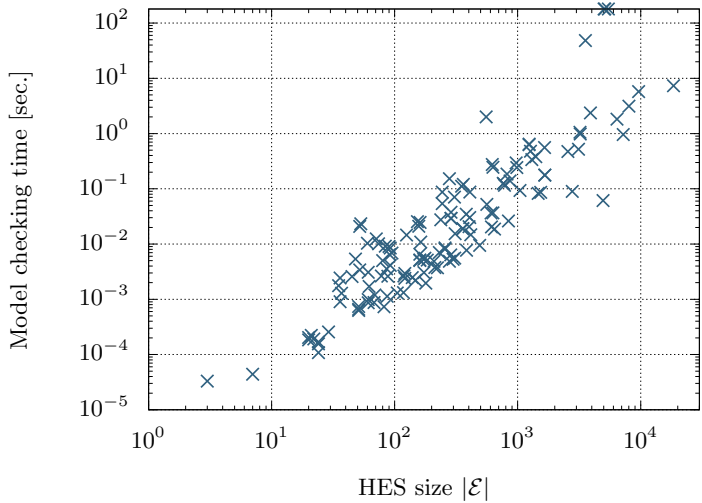}
    \caption{HES size $|\HES|$ versus time required for model checking ($\mbox{timeout} = \mbox{180 sec.}$)}
    \label{fig:ht}
\end{figure}

\section{Related Work}
\label{sec:rlwk}

The logic HFL has been introduced by Viswanathan and Viswanathan~\cite{Viswanathan-Viswanathan:2004}.
Later, Lange and his colleagues studied its various theoretical properties~\cite{Axelsson-Lange:2007,Axelson-Lange-Somla:2007,Lange-Lozes-Guzman:2014}.
In particular, they have shown that HFL model checking is \(k\)-EXPTIME complete for order \(k\) HFL formulas.
There has been, however, no practical HFL model checker.
Lozes has implemented a prototype HFL model checker,
but it is restricted to order-1 HFL,
and scales only for LTS of size up to 10 or so~\cite{LozesPC}.

Our algorithm is based on the type-based characterization of HFL model checking~\cite{Kobayashi-Lozes-Bruse:2017:POPL},
and type-based saturation algorithms for HORS model checking~\cite{Broadbent-Kobayashi:2013,Suzuki-Fujima-Kobayashi-Tsukada:2017}.
In particular,
the idea of restricting the arena of the typability game follows that of Suzuki et al.~\cite{Suzuki-Fujima-Kobayashi-Tsukada:2017}.
The detail of the algorithms are however different;
in particular,
the initial type environment in Suzuki et al.~\cite{Suzuki-Fujima-Kobayashi-Tsukada:2017}
contains \(F\COL\top\to\cdots \to \top \to q\) for any recursive function \(F\),
whereas in our algorithm,
\(\ITE_0\) contains \(F\COL\top\to\cdots \to \top \to q\) only for fixpoint variables bound by \(\nu\).
The use of a smaller initial type environment may be one of the reasons why our model checker tends to outperform theirs even for HORS model checking problems.
Another difference is in the correctness proofs.
In our opinion,
our proof is significantly simpler and streamlined than theirs.
Their proof manipulates infinite derivation trees.
In contrast,
our proof is a natural generalization of the correctness proof for the restricted fragment of HORS model checking
(which corresponds to the \(\mu\)-only or \(\nu\)-only fragment of HFL model checking)~\cite{Broadbent-Kobayashi:2013},
using the standard concept of parity progress measures.

\section{Conclusion}
\label{sec:cncl}

We have proposed the first practical algorithm for HFL model checking,
and proved its correctness.
We have confirmed through experiments that,
despite the huge worst-case complexity,
our prototype HFL model checker runs fast for typical instances of higher-order model checking.

\subsubsection*{Acknowledgments.}
We would like to thank anonymous referees for useful comments.
This work was supported by JSPS KAKENHI Grant Number JP15H05706.

\bibliographystyle{splncs04}
\bibliography{main}

\ifaplas\else

    \section*{Appendix}

    \appendix

    \section{Proofs for Section~\ref{sec:algo}}\label{app:proof:algo}

We only prove Lemma~\ref{lem:inverse},
since the other lemma (Lemma~\ref{lem:inverse2})
is not used in the proofs in Section~\ref{sec:prfs};
it was introduced just to explain the intuitions behind the algorithm.

First, we prove that typing is closed under the inverse of substitutions.

\begin{lemma} \label{invbase}
    Let $\fml$ and $\chi$ be HFL formulas containing no fixpoint operators or $\lambda$-abstractions.
    If $\ITE \vdash_{\LTS} [\chi / X] \, \fml \COL \ity$,
    then there exists a type environment $\Delta$
    such that $\dom(\Delta) \subseteq \FV(\fml) \cap \{ X \}$,
    $\ITE \cup \Delta \vdash_{\LTS} \fml \COL \ity$,
    and $\allprd{X \COL \ity' \in \Delta} \ITE \vdash_{\LTS} \chi \COL \ity'$.
\end{lemma}

\begin{proof}
    The proof proceeds by induction on the structure of the formula $\fml$.

    \begin{itemize}
        \item If $\fml = \true$/$\false$,
              then $\Delta = \emptyset$ satisfies the condition.

        \item If $\fml = Y$ and $Y \neq X$,
              then $\Delta = \emptyset$ satisfies the condition.

        \item If $\fml = X$,
              then $\Delta = \{ X \COL \ity \}$ satisfies the condition.

        \item If $\fml = \fml_{1} \lor \fml_{2}$,
              then $\ity = \state \in \States$,
              and we have either $\ITE \vdash_{\LTS} [\chi / X] \, \fml_{1} \COL \state$
              or $\ITE \vdash_{\LTS} [\chi / X] \, \fml_{2} \COL \state$.
              Suppose that $\ITE \vdash_{\LTS} [\chi / X] \, \fml_{i} \COL \state$.
              By the induction hypothesis,
              there exists $\Delta_{i}$
              such that $\dom(\Delta_{i}) \subseteq \FV(\fml_{i}) \cap \{ X \}$,
              $\ITE \cup \Delta_{i} \vdash_{\LTS} \fml_{i} \COL \state$,
              and $\allprd{X \COL \ity' \in \Delta_{i}} \ITE \vdash_{\LTS} \chi \COL \ity'$ hold.
              It is not difficult to see that this $\Delta_{i}$ satisfies the condition for $\Delta$.

        \item If $\fml = \fml_{1} \land \fml_{2}$,
              then $\ity = \state \in \States$,
              and we have both $\ITE \vdash_{\LTS} [\chi / X] \, \fml_{1} \COL \state$
              and $\ITE \vdash_{\LTS} [\chi / X] \, \fml_{2} \COL \state$.
              Therefore, by the induction hypothesis,
              there exists $\Delta_{i}$ for each $i = 1, 2$
              such that $\dom(\Delta_{i}) \subseteq \FV(\fml_{i}) \cap \{ X \}$,
              $\ITE \cup \Delta_{i} \vdash_{\LTS} \fml_{i} \COL \state$,
              and $\allprd{X \COL \ity' \in \Delta_{i}} \ITE \vdash_{\LTS} \chi \COL \ity'$ hold.
              Hence, $\Delta = \Delta_{1} \cup \Delta_{2}$ satisfies the condition.

        \item If $\fml = \diafml{a}{\fml'}$,
              then $\ity = \state \in \States$,
              and there exists $\state' \in \States$
              such that $\state \transto{a} \state'$
              and $\ITE \vdash_{\LTS} [\chi / X] \, \fml' \COL \state'$.
              Therefore, by the induction hypothesis,
              there exists $\Delta'$
              such that $\dom(\Delta') \subseteq \FV(\fml') \cap \{ X \}$,
              $\ITE \cup \Delta' \vdash_{\LTS} \fml' \COL \state'$,
              and $\allprd{X \COL \ity' \in \Delta'} \ITE \vdash_{\LTS} \chi \COL \ity'$ hold.
              Hence, $\Delta = \Delta'$ satisfies the condition.

        \item If $\fml = \boxfml{a}{\fml'}$,
              then $\ity = \state \in \States$,
              and for each $\state' \in \States$ with $\state \transto{a} \state'$,
              we have $\ITE \vdash_{\LTS} [\chi / X] \, \fml' \COL \state'$.
              Let $\States' = \{ \state' \in \States \mid \state \transto{a} \state' \} = \bigcup_{i \in I}\{ \state_{i} \}$.
              By the induction hypothesis,
              there is $\Delta_{i}$ for each $i \in I$
              such that $\dom(\Delta_{i}) \subseteq \FV(\fml') \cap \{ X \}$,
              $\ITE \cup \Delta_{i} \vdash_{\LTS} \fml' \COL \state_{i}$,
              and $\allprd{X \COL \ity' \in \Delta_{i}} \ITE \vdash_{\LTS} \chi \COL \ity'$ hold.
              Therefore, $\Delta = \bigcup_{i \in I}\Delta_{i}$ satisfies the condition.

        \item If $\fml = \fml_{1} \, \fml_{2}$,
              then there is an intersection type $\aty = \{ \ity_{1}, \ldots, \ity_{k} \}$
              satisfying $\ITE \vdash_{\LTS} [\chi / X] \, \fml_{1} \COL \aty \to \ity$
              and $\ITE \vdash_{\LTS} [\chi / X] \, \fml_{2} \COL \ity_{i}$ for each $i \in \{ 1, \ldots, k \}$.
              Therefore, by the induction hypothesis,
              there exists $\Delta_{0}$
              such that $\dom(\Delta_{0}) \subseteq \FV(\fml_{1}) \cap \{ X \}$,
              $\ITE \cup \Delta_{0} \vdash_{\LTS} \fml_{1} \COL \aty \to \ity$,
              and $\allprd{X \COL \ity' \in \Delta_{0}} \ITE \vdash_{\LTS} \chi \COL \ity'$.
              Moreover, for each $i \in \{ 1, \ldots, k \}$,
              there exists $\Delta_{i}$
              such that $\dom(\Delta_{i}) \subseteq \FV(\fml_{2}) \cap \{ X \}$,
              $\ITE \cup \Delta_{i} \vdash_{\LTS} \fml_{2} \COL \ity_{i}$,
              and $\allprd{X \COL \ity' \in \Delta_{i}} \ITE \vdash_{\LTS} \chi \COL \ity'$.
              Therefore, $\Delta = \bigcup_{i = 0}^{k}\Delta_{i}$ satisfies the condition.
    \end{itemize}\qed
\end{proof}

Now we are ready to prove Lemma~\ref{lem:inverse}.
\begin{proof}[Proof of Lemma~\ref{lem:inverse}]
    The proof is by induction on the derivation of $\fml \rewritesto_{\HES} \fml'$.
    For the base case,
    suppose that $\fml = F_{j} \, \chi_{1} \, \cdots \, \chi_{\ell}$
    and $\fml' = [\chi_{1} / X_{1}, \ldots, \chi_{\ell} / X_{\ell}] \, \psi_{j}$.
    By repeatedly applying Lemma~\ref{invbase},
    we can construct a type environment $\Delta$
    such that $\dom(\Delta) \subseteq \FV(\psi_{j}) \cap \{ X_{1}, \ldots, X_{\ell} \}$,
    $\ITE \cup \Delta \vdash_{\LTS} \psi_{j} \COL \state$,
    and $\allprd{X_{k} \COL \ity \in \Delta} \ITE \vdash_{\LTS} \chi_{k} \COL \ity$.
    Let $\aty_{k} = \{ \, \ity \mid X_{k} \COL \ity \in \Delta \, \}$.
    Since $\chi_{k} \in \Flow_{\HES}(X_{k})$ holds for every $k \in \{ 1, \ldots, \ell \}$,
    we have $F_{j} \COL \aty_{1} \to \cdots \to \aty_{\ell} \to \state \in \mathcal{F}_{\LTS, \HES}(\ITE)$
    by the definition of the function $\mathcal{F}_{\LTS, \HES}$.
    Therefore, by repeatedly applying the typing rule $(\textsc{T-App})$,
    we have $\mathcal{F}_{\LTS, \HES}(\ITE) \vdash_{\LTS} \fml \COL \state$.
    The induction steps are trivial.
    \qed
\end{proof}

    \section{Proofs for Section~\ref{sec:prfs}}\label{app:proof:prfs}

First, we see some properties of the truncated lexicographic order $\preceq_{k}$ on \ordinals{}.

\begin{lemma} \label{lem:ord}
    \[
        \begin{array}{rl}
            \mbox{\rm{(i).}} & \mbox{If $i \leq j$, then $\Beta \prec_{i} \Beta'$ implies that $\Beta \prec_{j} \Beta'$.} \\
            \mbox{\rm{(ii).}} & \mbox{If $i \leq j$, then $\Beta \preceq_{j} \Beta'$ implies that $\Beta \preceq_{i} \Beta'$.} \\
            \mbox{\rm{(iii).}} & \mbox{$\Beta(\ell) \prec_{j} \Beta$ and $\Beta(\ell) \preceq_{j - 1} \Beta$ hold for each $\Beta$ of length $j$.} \\
            \mbox{\rm{(iv).}} & \mbox{For a set of $\Beta$'s of length at most $n$, $\preceq_{n}$ gives a total order.}
        \end{array}
    \]
\end{lemma}

\begin{proof}
    Trivial.\qed
\end{proof}

The next lemma implies that,
if we start from the initial \ordinal{} $\Beta_{0} = (m)$
and inductively define $\Beta_{i} = \Beta_{i - 1}(\ell_{i})$
by a sequence of positive integers $\ell_{1}, \, \ell_{2}, \, \ldots$
until it reaches a final \ordinal{} $\Beta_{k} = (\beta_{1}, \ldots, \beta_{j - 1}, 0)$,
then the sequence $\Beta_{0}, \, \ldots, \, \Beta_{k}$
contains all $(\beta_{1}, \ldots, \beta_{j - 1}, i)$'s for $i \in \{0, \ldots, m\}$.

\begin{lemma} \label{lem:beta}
    Let $\Beta_{0}, \, \Beta_{1}, \, \ldots, \, \Beta_{k}$ be a sequence of \ordinals{} such that,
    for each $i \in \{1, \ldots, k\}$, $\Beta_{i} = \Beta_{i - 1}(\ell_{i})$ holds for some $\ell_{i}$.
    If $\Beta_{k} = (\beta_{1}, \ldots, \beta_{j - 1}, \beta_{j})$, $\beta_{j} < m$,
    and $(\beta_{1}, \ldots, \beta_{j - 1}, \beta_{j} + 1) \preceq_{j} \Beta_{0}$,
    then there exists $i \in \{0, \ldots, k - 1\}$ such that
    $\Beta_{i} = (\beta_{1}, \ldots, \beta_{j - 1}, \beta_{j} + 1)$.
\end{lemma}

\begin{proof}
    The proof proceeds by induction on the number of successors $k$.
    For the base case $k = 1$,
    suppose that $\Beta_{0}(\ell) = (\beta_{1}, \ldots, \beta_{j - 1}, \beta_{j})$.
    By the definition of the operation $\Beta(\ell)$ and the assumption $\beta_{j} < m$,
    there are only two possibilities:
    (i) $\Beta_{0} = (\beta_{1}, \ldots, \beta_{j - 1}, \beta_{j} + 1)$ and $\ell = j$, or
    (ii) $\Beta_{0} = (\beta_{1}, \ldots, \beta_{j - 1}, \beta_{j}, \ldots, \beta_{j'})$ and $\ell = j < j'$.
    Since we also have $(\beta_{1}, \ldots, \beta_{j - 1}, \beta_{j} + 1) \preceq_{j} \Beta_{0}$,
    (i) must be the case, and thus $\Beta_{0} = (\beta_{1}, \ldots, \beta_{j - 1}, \beta_{j} + 1)$ holds.

    For the step case,
    if $(\beta_{1}, \ldots, \beta_{j - 1}, \beta_{j} + 1) \preceq_{j} \Beta_{1}$,
    then we are done by the induction hypothesis.
    Suppose otherwise,
    and let $\ell_{0}$ and $\ell_{1}$ be the lengths of $\Beta_{0}$ and $\Beta_{1}$, respectively.
    Then, since
    $\Beta_{k} = (\beta_{1}, \ldots, \beta_{j - 1}, \beta_{j}) \preceq_{j} \Beta_{1} \prec_{j} (\beta_{1}, \ldots, \beta_{j - 1}, \beta_{j} + 1)$,
    we have $\Beta_{1} =_{j} (\beta_{1}, \ldots, \beta_{j - 1}, \beta_{j})$, and thus $j \leq \ell_{1}$ holds.
    If $\ell_{0} < j$, then we have $\beta_{j} = m$, and thus it contradicts the assumption $\beta_{j} < m$.
    Moreover, if $j < \ell_{0}$, then $\Beta_{0} =_{j} \Beta_{1}$ must hold since we also have $j \leq \ell_{1}$.
    This contradicts the assumption
    $\Beta_{1} \prec_{j} (\beta_{1}, \ldots, \beta_{j - 1}, \beta_{j} + 1) \preceq_{j} \Beta_{0}$,
    and thus we have $\ell_{0} = j$.
    Hence, it follows that $\Beta_{0} = (\beta_{1}, \ldots, \beta_{j - 1}, \beta_{j} + 1)$
    from $\Beta_{1} =_{j} (\beta_{1}, \ldots, \beta_{j - 1}, \beta_{j})$.
    \qed
\end{proof}

\subsection{Proof of Lemma~\ref{lem:regress}}

The following is a key lemma,
which can be seen as defining a \emph{parity progress measure}~\cite{Jurdzinski:2000} for typability games.

\begin{lemma}[Parity Progress Measure for Typability Games] \label{lem:progress}
    First, let $\TG(\LTS, \HES) = (V_{0}, V_{1}, v_{0}, E, \Omega)$ be a typability game.
    Let $W_{0}$ be a subset of $V_{0}$
    and $\strategy \COL W_{0} \to V_{1}$ be a strategy on $W_{0}$ such that
    $\strategy(F_{j} \COL \ity) \subseteq W_{0}$ holds for every $F_{j} \COL \ity \in W_{0}$.
    If we can assign to each $F_{j} \COL \ity \in W_{0}$
    an \ordinal{} $\xi(F_{j} \COL \ity)$ satisfying the following conditions:
    \[
        \begin{array}{rll}
            (1). & \allprd{F_{j'} \COL \ity' \in \strategy(F_{j} \COL \ity)} \xi(F_{j'} \COL \ity') \prec_{j} \xi(F_{j} \COL \ity) & (\mbox{if }\fpo_{j} = \mu) \\
            (2). & \allprd{F_{j'} \COL \ity' \in \strategy(F_{j} \COL \ity)} \xi(F_{j'} \COL \ity') \preceq_{j - 1} \xi(F_{j} \COL \ity) & (\mbox{if }\fpo_{j} = \nu),
        \end{array}
    \]
    then $\strategy$ is a winning strategy of \verifier{} for each $F_{j} \COL \ity \in W_{0}$
    in the typability game $\TG(\LTS, \HES)$,
    and also of the subgame \(\SG(\LTS,\HES,W_0)\).
\end{lemma}

\begin{proof}
    First, note that every finite maximal play
    starting from $v \in W_{0}$ and conforming with $\strategy$
    is winning for \verifier{} since $\strategy$ is total.
    Moreover, since $\Omega(\ITE) = 0$ for each $\ITE \in V_{1}$,
    whether an infinite play satisfies the parity condition can be checked
    just by looking at the priorities of the positions in $V_{0}$.

    Suppose that $\xi$ satisfies the conditions (1) and (2).
    Let $G = (W_{0}, E_{\strategy})$ be a directed graph such that
    $(v, w) \in E_{\strategy}$ iff $w \in \strategy(v)$.
    Since every play starting from $v \in W_{0}$
    and conforming with $\strategy$ corresponds to a walk in $G$
    (by ignoring all positions in \(V_{1}\)),
    the strategy $\strategy$ gives a winning strategy of
    \verifier{} for each $v \in W_{0}$
    if every cycle in $G$ is even.
    Here we say that a cycle is \emph{even}
    if the highest priority of its vertices is even,
    and otherwise the cycle is said to be \emph{odd}.

    Suppose that there is an odd cycle in $G$,
    and let $\mathcal{C}$ be such a cycle.
    Let $F_{j} \COL \ity$ be a vertex in $\mathcal{C}$
    that has the minimum $j$ and thus the highest (odd) priority.
    Then, the cycle $\mathcal{C}$ can be written as
    $F_{j} \COL \ity = v_{0} v_{1} \cdots v_{\ell} = F_{j} \COL \ity$.
    Let $v_{k} = F_{j'} \COL \ity'$ be a vertex in $\mathcal{C}$ with $0 \leq k < \ell$.
    If $\fpo_{j'} = \mu$,
    then $\xi(v_{k + 1}) \prec_{j'} \xi(v_{k})$ holds by the condition (1).
    Moreover, since $j \leq j'$,
    we have $\xi(v_{k + 1}) \preceq_{j} \xi(v_{k})$ by Lemma~\ref{lem:ord} (ii).
    If $\fpo_{j'} = \nu$,
    then $\xi(v_{k + 1}) \preceq_{j' - 1} \xi(v_{k})$ holds by the condition (2),
    and since $j \leq j' - 1$,
    we have $\xi(v_{k + 1}) \preceq_{j} \xi(v_{k})$ again by Lemma~\ref{lem:ord} (ii).
    Therefore,
    $\xi(v_{0}) \succ_{j} \xi(v_{1}) \succeq_{j} \cdots \succeq_{j} \xi(v_{\ell}) = \xi(v_{0})$ holds.
    This is a contradiction, and thus every cycle in $G$ is even.
    Hence, $\strategy$ is a winning strategy of \verifier{} for each $v \in W_{0}$.
    Moreover, since $\strategy(F_{j} \COL \ity) \subseteq W_{0}$ holds for each \(F_{j} \COL \ity \in W_{0}\),
    it is also a winning strategy for the subgame \(\SG(\LTS,\HES,W_0)\).
    \qed
\end{proof}

We are now ready to prove Lemma~\ref{lem:regress}.

\begin{proof}[Proof of Lemma~\ref{lem:regress}]
    Let $W = \Forget(\ITE)$.
    To each $F_{j} \COL \ity \in W$,
    we assign an \ordinal{} $\xi(F_{j} \COL \ity)$ by
    $\xi(F_{j} \COL \ity) = \min_{\preceq_{n}} \{ \Beta \mid F_{j}^{\Beta} \COL \ity \in \ITE \}$.
    By the definition of the function $\Regress$,
    for each $F_{j} \COL \ity \in W$,
    there exists a type environment
    $\ITE' \subseteq \ITE_{\fpo_{j}}^{\xi(F_{j} \COL \ity)}$ such that
    $\ITE' \vdash_{\LTS} \fml_{j} \COL \ity$ holds.
    We define $\strategy(F_{j} \COL \ity)$ as such $\ITE'$.
    Then, the function $\strategy$ becomes a strategy on $W$
    satisfying $\strategy(v) \subseteq W$ for every $v \in W$.

    Let $F_{j'} \COL \ity'$ be a type binding in $\strategy(F_{j} \COL \ity)$.
    If $\fpo_{j} = \mu$,
    then there exists $F_{j'}^{\Beta'} \COL \ity' \in \ITE$
    satisfying $\Beta' \prec_{j} \xi(F_{j} \COL \ity)$
    since $\strategy(F_{j} \COL \ity) \subseteq \ITE_{\fpo_{j}}^{\xi(F_{j} \COL \ity)}$ holds.
    Moreover, we have $\xi(F_{j'} \COL \ity') \preceq_{n} \Beta'$
    by the definition of $\xi$,
    and thus $\xi(F_{j'} \COL \ity') \preceq_{j} \Beta'$ by Lemma~\ref{lem:ord} (ii).
    Hence $\xi(F_{j'} \COL \ity') \prec_{j} \xi(F_{j} \COL \ity)$ holds.
    Similarly, if $\fpo_{j} = \nu$,
    then we have $\xi(F_{j'} \COL \ity') \preceq_{j - 1} \xi(F_{j} \COL \ity)$.
    Therefore, $W$ consists of only winning positions of $\SG(\LTS, \HES, W)$ by Lemma~\ref{lem:progress},
    hence also of $\SG(\LTS, \HES, \Forget((\mathcal{F}_{\LTS, \HES^{(m)}})^{\omega}(\emptyset)))$.
    \qed
\end{proof}

\subsection{Proofs for Lemma~\ref{lem:regress2}}

Next, we prove the lemmas that are required to prove Lemma~\ref{lem:regress2}.

\begin{lemma} \label{lem:inverse3}
    If \(\LTS \models \HES\),
    then $F_{1}^{(m)} \COL q_{0} \in \ITEomega{m}$ holds for sufficiently large $m$.
\end{lemma}

\begin{proof}
    For sufficiently large \(m\),
    \(\LTS\models \HES\) if and only if \(\LTS\models \HES^{(m)}\).
    Pick such \(m\).
    Let $\fml^{(m)}$ be a formula such that
    $F_{1}^{(m)} \rewritesto_{\HES^{(m)}}^{*} \fml^{(m)} \nrewritesto_{\HES^{(m)}}$.
    Since \(\LTS\models \HES^{(m)}\) and the reduction preserves the semantics,
    we have $\emptyset \vdash_{\LTS} \fml^{(m)} \COL q_{0}$.
    Hence, we obtain $\ITEomega{m} \vdash_{\LTS} F_{1}^{(m)} \COL q_{0}$
    and thus $F_{1}^{(m)} \COL q_{0} \in \ITEomega{m}$
    by repeatedly applying Lemma~\ref{lem:inverse}.
    \qed
\end{proof}

Given an HES $\HES = \HESRHS$
and a type environment $\ITE$ with $\dom(\ITE) \subseteq \{F_{1}, \ldots, F_{n}\}$,
we write $\vdash_{\LTS} \HES \COL \ITE$
when $\ITE \vdash_{\LTS} \fml_{j} \COL \ity$ holds for every $F_{j} \COL \ity \in \ITE$.

\begin{lemma} \label{lem:preserve}
    If $\vdash_{\LTS} \HES \COL \ITE$,
    then $\vdash_{\LTS} \HES \COL \mathcal{F}_{\LTS, \HES}(\ITE)$.
\end{lemma}

\begin{proof}
    Suppose that $\vdash_{\LTS} \HES \COL \ITE$.
    Let $F_{j} \COL \ity$ be a type binding with $F_{j} \COL \ity \in \mathcal{F}_{\LTS, \HES}(\ITE)$,
    where $\ity = \aty_{1} \to \cdots \to \aty_{\ell} \to \state$
    and $\fml_{j} = \lmdfml{X_{1}}{} \cdots \lmdfml{X_{\ell}}{} \psi_{j}$.
    If $F_{j} \COL \ity \in \ITE$,
    then, we have $\ITE \vdash_{\LTS} \fml_{j} \COL \ity$ since $\vdash_{\LTS} \HES \COL \ITE$,
    and thus $\mathcal{F}_{\LTS, \HES}(\ITE) \vdash_{\LTS} \fml_{j} \COL \ity$
    holds by $\ITE \subseteq \mathcal{F}_{\LTS, \HES}(\ITE)$.
    If $F_{j} \COL \ity \notin \ITE$,
    then, by the definition of the function $\mathcal{F}_{\LTS, \HES}$,
    there exists a type environment $\Delta$ such that
    $\dom(\Delta) \subseteq \FV(\psi_{j}) \cap \{ X_{1}, \ldots, X_{\ell} \}$,
    $\ITE \cup \Delta \vdash_{\LTS} \psi_{j} \COL \state$,
    and $\allprd{k \in \{ 1, \ldots, \ell \}} \aty_{k} = \Delta(X_{k})$.
    Therefore, by repeatedly applying the typing rule $(\textsc{T-Abs})$,
    we have $\ITE \vdash_{\LTS} \fml_{j} \COL \ity$,
    and thus $\mathcal{F}_{\LTS, \HES}(\ITE) \vdash_{\LTS} \fml_{j} \COL \ity$.
    \qed
\end{proof}

\begin{lemma} \label{lem:preserve2}
    Let $F_{j}^{\Beta} \COL \ity$ be a type binding in $\ITEomega{m}$.
    If $\beta_{j} = 0$,
    then we have $\emptyset \vdash_{\LTS} \varphi_{j}^{\Beta} \COL \ity$.
    Otherwise,
    we have $\ITEomega{m} \downarrow_{\{F_{1}^{\Beta(1)}, \ldots, F_{n}^{\Beta(n)}\}} \vdash_{\LTS} \fml_{j}^{\Beta} \COL \ity$,
    where $\ITE \downarrow_{S}$ denotes the restriction of $\ITE$ by $S$,
    i.e., $\ITE \downarrow_{S} = \{ F \COL \ity \in \ITE \mid F \in S \, \}$.
\end{lemma}

\begin{proof}
    Since $\vdash_{\LTS} \HES^{(m)} \COL \emptyset$ is vacuously true,
    we have $\vdash_{\LTS} \HES^{(m)} \COL \ITEomega{m}$ by repeatedly applying Lemma~\ref{lem:preserve}.
    Therefore, $\ITEomega{m} \vdash_{\LTS} \fml_{j}^{\Beta} \COL \ity$ holds for each $F_{j}^{\Beta} \COL \ity \in \ITEomega{m}$.
    Let $F_{j}^{\Beta} \COL \ity$ be a type binding in $\ITEomega{m}$.
    Suppose that $\beta_{j} = 0$.
    Then, $\alpha_{j}$ must be $\nu$ because otherwise
    $\fml_{j}^{\Beta} = \lmdfml{\widetilde{X}}{} \false$ cannot be typed by any type environment.
    Therefore, $\fml_{j}^{\Beta}$ is of the form $\lmdfml{\widetilde{X}}{} \true$
    and thus we have $\emptyset \vdash_{\LTS} \varphi_{j} \COL \ity$.
    When $\beta_{j} \neq 0$,
    we have $\FV(\fml_{j}^{\Beta}) \subseteq \{ F_{1}^{\Beta(1)}, \ldots, F_{n}^{\Beta(n)} \}$
    by the definition of the formula $\fml_{j}^{\Beta}$, and thus
    $\ITEomega{m} \downarrow_{\{F_{1}^{\Beta(1)}, \ldots, F_{n}^{\Beta(n)}\}} \vdash_{\LTS} \fml_{j}^{\Beta} \COL \ity$ always holds.
    \qed
\end{proof}

We now restate and prove Lemmas~\ref{lem:seqbase} and~\ref{lem:seqstep}.

\repeatlemma{lem:seqbase}

\begin{proof}
    Let $F_{j}^{\Beta} \COL \ity$ be a type binding with
    $\beta_{j} = 0$ and $\fpo_{j} = \mu$.
    Then, $\fml_{j}^{\Beta} = \lmdfml{\widetilde{X}}{}\false$
    by the definition of the formula $\fml_{j}^{\Beta}$.
    If $F_{j}^{\Beta} \COL \ity \in \ITEomega{m}$,
    then $\ITEomega{m} \vdash_{\LTS} \fml_{j}^{\Beta} \COL \ity$
    holds by Lemma~\ref{lem:preserve2}.
    However, since $\fml_{j}^{\Beta}$ cannot be typed by any type environments,
    we have $\ITEomega{m} \nvdash_{\LTS} \fml_{j}^{\Beta} \COL \ity$.
    Therefore, we have $F_{j}^{\Beta} \COL \ity \notin \ITEomega{m}$
    and thus $F_{j}^{\Beta} \COL \ity \notin D_{k}$ for any $k$.
    \qed
\end{proof}

\repeatlemma{lem:seqstep}

\begin{proof}
    Let $F_{j}^{\Beta} \COL \ity$ be a type binding satisfying
    $F_{j}^{\Beta} \COL \ity \in D_{k}$ and $\beta_{j} \neq 0$.
    Since $D_{k} \subseteq \ITEomega{m}$,
    we have $F_{j}^{\Beta} \COL \ity \in \ITEomega{m}$.
    Therefore, by Lemma~\ref{lem:preserve2},
    there exists a type environment $\ITE$
    satisfying $\ITE \subseteq \ITEomega{m} \downarrow_{\{F_{1}^{\Beta(1)}, \ldots, F_{n}^{\Beta(n)}\}}$
    and $\ITE \vdash_{\LTS} \fml_{j}^{\Beta} \COL \ity$.
    For such $\ITE$,
    we have $\Forget(\ITE) = \ITE_{\fpo_{j}}^{\Beta}$ by Lemma~\ref{lem:ord} (iii).
    In addition,
    we have $\fml_{j}^{\Beta} = \lbrack F_{1}^{\Beta(1)}/F_{1}, \ldots, F_{n}^{\Beta(n)}/F_{n} \rbrack \,\fml_{j}$
    since $\beta_{j} \neq 0$.
    Hence $\ITE_{\fpo_{j}}^{\Beta} \vdash_{\LTS} \fml_{j} \COL \ity$ holds.
    Moreover,
    we have $(\Regress^{k - 1}(\ITEomega{m}))_{\fpo_{j}}^{\Beta} \nvdash_{\LTS} \fml_{j} \COL \ity$
    since $F_{j}^{\Beta} \COL \ity \in D_{k}$.
    Therefore, $\ITE \nsubseteq \Regress^{k - 1}(\ITEomega{m})$ holds,
    and thus there must exist a type binding
    $F_{j'}^{\Beta(j')} \COL \ity' \in \ITE$ such that
    $F_{j'}^{\Beta(j')} \COL \ity' \notin \Regress^{k - 1}(\ITEomega{m})$.
    This implies that $F_{j'}^{\Beta(j')} \COL \ity' \in D_{k'}$
    for some $k'$ satisfying $1 \leq k' < k$.
    \qed
\end{proof}

\subsection{Proof of Lemma~\ref{lem:overapproximation-of-Em}}

Lemma~\ref{lem:overapproximation-of-Em} is an immediate corollary of the following lemma.

\begin{lemma} \label{lem:subseteq}
    Let $\ITE$ and $\ITE'$ be type environments for $\HES$ and $\HES^{(m)}$, respectively.
    If $\ITE_{0} \subseteq \ITE$ and $\Forget(\ITE') \subseteq \ITE$,
    then $\Forget(\mathcal{F}_{\LTS, \HES^{(m)}}(\ITE')) \subseteq \mathcal{F}_{\LTS, \HES}(\ITE)$.
\end{lemma}

\begin{proof}
    Let $\ITE$ and $\ITE'$ be type environments for $\HES$ and $\HES^{(m)}$, respectively,
    such that both $\ITE_{0} \subseteq \ITE$ and $\Forget(\ITE') \subseteq \ITE$ hold.
    Let $F_{j}^{\Beta} \COL \ity \in \mathcal{F}_{\LTS, \HES^{(m)}}(\ITE') \backslash \ITE'$.
    We show that $F_{j} \COL \ity \in \mathcal{F}_{\LTS, \HES}(\ITE)$.
    Let $\ity = \aty_{1} \to \cdots \to \aty_{\ell} \to q$,
    $\fml_{j} = \lmdfml{X_{1}}{} \cdots \lmdfml{X_{\ell}}{} \psi_{j}$,
    and $\fml_{j}^{\Beta} = \lmdfml{X_{1}}{\Beta} \cdots \lmdfml{X_{\ell}}{\Beta} \psi_{j}^{\Beta}$.
    By the definition of the function $\mathcal{F}_{\LTS, \HES^{(m)}}$,
    there is a type environment $\Delta'$ such that
    $\dom(\Delta') \subseteq \FV(\psi_{j}^{\Beta}) \cap \{ X_{1}^{\Beta}, \ldots, X_{\ell}^{\Beta} \}$,
    $\ITE' \cup \Delta' \vdash_{\LTS} \psi_{j}^{\Beta} \COL q$,
    $\allprd{k \in \{ 1, \ldots, \ell \}} \aty_{k} = \Delta'(X_{k}^{\Beta})$,
    and $\allprd{X_{k}^{\Beta} \in \dom(\Delta')} \extprd{\fml \in \Flow_{\HES^{(m)}}(X_{k}^{\Beta})}$
    $\allprd{\ity' \in \Delta'(X_{k}^{\Beta})} \ITE' \vdash_{\LTS} \fml \COL \ity'$.

    Suppose that $\beta_{j} = 0$.
    If $\fpo_{j} = \mu$,
    then $\psi_{j}^{\Beta} = \false$
    and thus it cannot be typed by any type environments.
    Therefore, $\fpo_{j}$ must be $\nu$.
    Since $\FV(\psi_{j}^{\Beta})$ is empty,
    $\Delta' = \emptyset$, and thus
    $F_{j} \COL \ity = F_{j} \COL \top \to \cdots \to \top \to q \in \ITE_{0} \subseteq \ITE \subseteq \mathcal{F}_{\LTS, \HES}(\ITE)$.

    Suppose that $\beta_{j} \neq 0$,
    and let $\Delta = \Forget(\Delta')$.
    Let us write $\Forget(\fml)$ for the formula obtained by removing all indices from $\fml$.
    Then, since $\Forget(\psi_{j}^{\Beta}) = \psi_{j}$
    and $\Forget(\ITE' \cup \Delta') = \Forget(\ITE') \cup \Delta \subseteq \ITE \cup \Delta$,
    we have $\dom(\Delta) \subseteq \FV(\psi_{j}) \cap \{X_{1}, \ldots, X_{\ell}\}$
    and $\ITE \cup \Delta \vdash_{\LTS} \psi_{j} \COL q$.
    Moreover, since $\fml \in \Flow_{\HES^{(m)}}(X_{k}^{\Beta})$
    implies that $\Forget(\fml) \in \Flow_{\HES}(X_{k})$,
    it follows that $\allprd{X_{k} \in \dom(\Delta)} \extprd{\fml \in \Flow_{\HES}(X_{k})}$
    $\allprd{\ity' \in \Delta(X_{k})} \ITE \vdash_{\LTS} \fml \COL \ity'$.
    Therefore, $F_{j} \COL \ity \in \mathcal{F}_{\LTS, \HES}(\ITE)$
    holds by the definition of the function $\mathcal{F}_{\LTS, \HES}$.
    \qed
\end{proof}

\begin{proof}[Proof of Lemma~\ref{lem:overapproximation-of-Em}]
    It suffices to show
    \(\Forget((\mathcal{F}_{\LTS, \HES^{(m)}})^{i}(\emptyset)) \subseteq (\mathcal{F}_{\LTS, \HES})^{i}(\ITE_0)\)
    for any \(i\).
    We prove it by induction on \(i\).
    The base case where \(i=0\) is trivial.
    In the case where \(i>0\), we have
    \(\Forget((\mathcal{F}_{\LTS, \HES^{(m)}})^{i-1}(\emptyset)) \subseteq (\mathcal{F}_{\LTS, \HES})^{i-1}(\ITE_0)\)
    by the induction hypothesis.
    Since \(\ITE_0 \subseteq (\mathcal{F}_{\LTS, \HES})^{i-1}(\ITE_0)\),
    we have \(\Forget((\mathcal{F}_{\LTS, \HES^{(m)}})^{i}(\emptyset)) \subseteq (\mathcal{F}_{\LTS, \HES})^{i}(\ITE_0)\)
    by Lemma~\ref{lem:subseteq}.
    \qed
\end{proof}

    \section{Several Optimizations}
\label{app:opt}

In the following,
we report on several optimization techniques that are adopted in our prototype model checker \HFLMC{}
and thought to have improved its performance in the experiments in Section~\ref{sec:expr}.

\subsection{Flow Analysis with Type Information}
\label{subsec:flowcomputation}

The function \(\mathit{F}_{\LTS, \HES}\)
expands type environments in a backward direction
of the rewriting sequence starting from the initial variable \(F_{1}\).
The objective of this expansion is to gather sufficient type bindings
to derive the initial position \(F_{1} \COL \state_{0}\) of the typability game,
and thus what we really need is to expand type environments
in a backward direction of the \emph{derivation} of the initial position.
This fact enables us to restrict the expansion using information on types.

For example, if the initial equation of the HES is
\(F_{1} =_{\fpo} \diafml{a}A \lor \diafml{b}B\),
then \(\Gamma \vdash_{\LTS} \fml_{1} \COL \state_{0}\) holds if and only if we have either
\(A \COL \state \in \Gamma\) for some \(\state\) with \(\state_{0} \transto{a} \state\)
or \(B \COL \state \in \Gamma\) for some \(\state\) with \(\state_{0} \transto{b} \state\).
If \(F_{1}\), \(A\), and \(B\) do not occur in the body of any fixpoint variable other than \(F_{1}\),
then we can safely exclude
\(F_{1} \COL \state\) with \(q \neq \state_{0}\),
\(A \COL \state\) with \(\state_{0} \nottransto{a} \state\),
and \(B \COL \state\) with \(\state_{0} \nottransto{b} \state\) from consideration.
Although we do not describe it in detail here,
we can gather such information on types during flow analysis
using the idea of abstract configuration graphs~\cite{Ramsay-Neatherway-Ong:2014:POPL}
adapted to the context of HFL model checking.

\subsection{Restriction on the Initial Type Environment}
\label{subsec:initialtypeenv}

In Algorithm~\ref{lst:pseudocode},
the initial type environment \(\Gamma_{0}\)
contains the strongest type bindings for all \(\nu\)-variables.
However, as the completeness proof in Appendix~\ref{app:proof:prfs} indicates,
it is sufficient to have the strongest type bindings for the
\(\nu\)-variables approximated as
\(\lmdfml{\widetilde{X}}{}\widehat{\nu} = \lmdfml{\widetilde{X}}{}\true\)
in a formula \(\fml\) such that \(F_{1}^{(m)} \rewritesto_{\HES^{(m)}}^{*} \fml\).

Such variables can be found by analysing a \emph{call graph} of the HES.
A call graph \(G = (V, E)\) of an HES \(\HES\) is a directed graph
whose vertex set \(V\) is the set of fixpoint variables of \(\HES\),
and \((F_{j}, F_{k}) \in E\) holds if and only if the variable \(F_{k}\)
occurs in the body of the variable \(F_{j}\).
It is sufficient for the initial type environment \(\Gamma_{0}\)
to contain the strongest type bindings for the \(\nu\)-variables \(F\) such that
there is a cycle in \(G\) whose highest priority is given by \(F\).\footnote{
    Moreover, we can restrict the use of such strongest type bindings
    to the type derivations for the fixpoint variables that call \(F\) in such a cycle.
}

\subsection{Normalization through Subtyping}
\label{subsec:normalizesub}

The objective of backward expansions by the function \(\mathit{F}_{\LTS, \HES}\)
is to construct a subgame that is winning for \verifier{}
whenever the full typability game is winning for her.
The actual implementation of the expansion process reflects this view,
and it constructs not only vertices but at the same time edges of the arena of a subgame.
Some edges (moves) of a typability game are comparable in terms of their ``strengths'' from the viewpoint of \verifier{},
and this fact enables us to decrease the size of subgames and possibly reduce the computational cost of expansion processes.

For an illustration,
let us consider the situation in which we are to derive new types for a variable \(F\),
whose body is \(\fml_{F} = \lmdfml{X}{\obj} \appfml{G} (\appfml{G} X)\),
in the expansion process of the type environment
\(\Gamma = \{G \COL \top \to \state, G \COL \state \to \state\}\).
Then, we can derive \(F \COL \aty \to \state\) for any \(\aty\)
because \(\Gamma\) contains a strongest type binding \(G \COL \top \to \state\)
which makes the derivation
\(\{G \COL \top \to \state\} \vdash_{\LTS} \fml_{F} \COL \aty \to \state\) possible.

However, we can also derive \(\fml_{F} \COL \state \to \state\)
from the type environment \(\{G \COL \state \to \state\}\),
and the edge from \(F \COL \state \to \state\)
to \(\{G \COL \state \to \state\}\) is always stronger than
the edge from \(F \COL \state \to \state\) to \(\{G \COL \top \to \state\}\) in a typability game;
if \(\{G \COL \top \to \state\}\) is a winning position of \verifier{},
then so is \(\{G \COL \state \to \state\}\),
because \(\top \to \state \leq_{\LTS} \state \to \state\)
and thus \(\ITE' \vdash_{\LTS} \fml_{G} \COL \top \to \state\)
implies \(\ITE' \vdash_{\LTS} \fml_{G} \COL \state \to \state\).
Hence, \verifier{} can safely choose
\(\{G \COL \state \to \state\}\) instead of
\(\{G \COL \top \to \state\}\) from \(F \COL \state \to \state\).
As a result, we can refrain from constructing
the edge from \(F \COL \state \to \state\) to \(\{G \COL \top \to \state\}\)
without undermining the completeness of the algorithm.

Moreover, although \(F \COL \aty \to \state\) is derivable for any \(\aty\),
we can refrain from adding type bindings other than
\(F \COL \top \to \state\) and \(F \COL \state \to \state\)
to \(\Gamma\) at this point of time.
This is because for any \(\Gamma'\) and \(\aty\) such that
\(\Gamma' \subseteq \Gamma\),
\(\aty \neq \emptyset, \{\state\}\),
and \(\Gamma' \vdash_{\LTS} \varphi_{F} \COL \aty \to \state\),
we also have \(\Gamma' \vdash_{\LTS} \varphi_{F} \COL \top \to \state\)
and the edge from \(F \COL \aty \to \state\) to \(\Gamma'\)
is always weaker than the edge from \(F \COL \top \to \state\) to \(\Gamma'\);
if there is a strategy \(\strategy\) with \(\strategy(F \COL \aty \to \state) = \Gamma'\)
that gives a winning strategy for \verifier{} from \(F \COL \aty \to \state\),
then the position \(F \COL \top \to \state\) is also winning,
and thus \verifier{} can safely use \(F \COL \top \to \state\)
instead of \(F \COL \aty \to \state\)
by exploiting the relation \(\top \to \state \leq_{\LTS} \aty \to \state\).

In general,
when the body of a fixpoint variable \(F\) is of the form \(\lmdfml{\widetilde{X}}{}\psi_{F}\),
an edge obtained from a type derivation \(\Gamma \vdash_{\LTS} \psi_{F} \COL \state\)
is stronger than an edge obtained from a derivation \(\Gamma' \vdash_{\LTS} \psi_{F} \COL \state\)
if and only if \(\Gamma\) is weaker as a type environment than \(\Gamma'\)
(that is, we have \(\allprd{X \in \dom(\Gamma)} \Gamma'(X) \leq_{\LTS} \Gamma(X)\)).
In our prototype model checker \HFLMC{},
we use the inclusion relation \(\Gamma \subseteq \Gamma'\)
instead of \(\allprd{X \in \dom(\Gamma)} \Gamma'(X) \leq_{\LTS} \Gamma(X)\)
as it is computationally cheaper,
although the relations are not equivalent and it requires further analysis to decide
if the choice actually improves the performance of the model checker in general cases.

    \section{Content of the Benchmark Set} \label{sec:bench}

In the following,
we report on some details on the content of the benchmark set
we used in the experiments in Section~\ref{sec:expr}.\footnote{
    It is available at \url{https://github.com/hopv/homusat/tree/master/bench/exp}.
}

Figure~\ref{fig:ol} shows a plot of the order of HES versus the size of LTS.
As can be seen, the LTS size is moderate (up to around 10)
for most of the instances, but some instances have large state sets.
For example, \texttt{xhtmlf-div-2},
which was taken from the \textsc{HorSat2} benchmark set,
consists of an HES of order-2 and an LTS of size 220.
Although for such an instance the possible number of types is enormously huge,
\HFLMC{} solved it in fewer than 2 seconds.

Figure~\ref{fig:hist} shows the distribution of
the numbers of alternations between \(\mu\) and \(\nu\) within the benchmark set.
Whereas over half of the instances (83 out of 136) have no alternation,
there are a certain number of instances that have one or more alternations, up to a maximum of 4.
Note that,
although the number of alternations is one of the factors that affect the difficulty of input problems,
it is not so significant at least for the problems we used.
Actually, all three problems that \HFLMC{} timed out
(\texttt{file-e}, \texttt{file\textasciitilde{}2}, and \texttt{intro})
have no alternation.
All of them are \(\nu\)-only problems with order 3,
\(|Q| = 4\), and \(|\HES| \approx 5,000\),
taken from \textsc{HorSat2} benchmarks.

\begin{figure}[t]
    \input{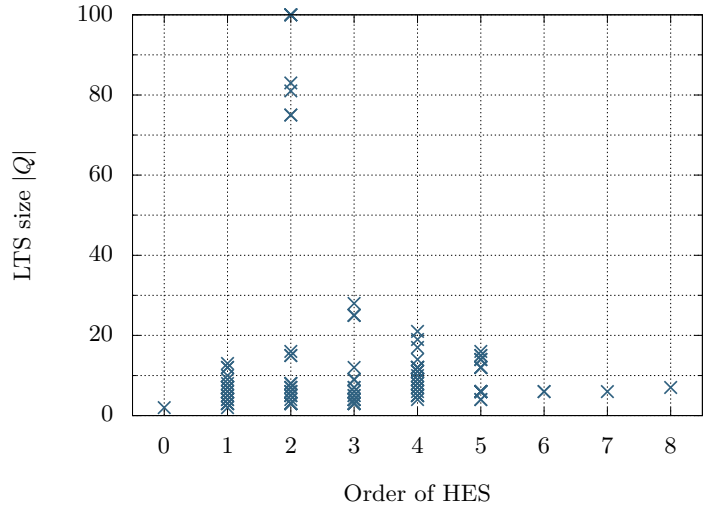}
    \caption{Order of HES versus LTS size $|Q|$}
    \label{fig:ol}
\end{figure}
\begin{figure}
    \input{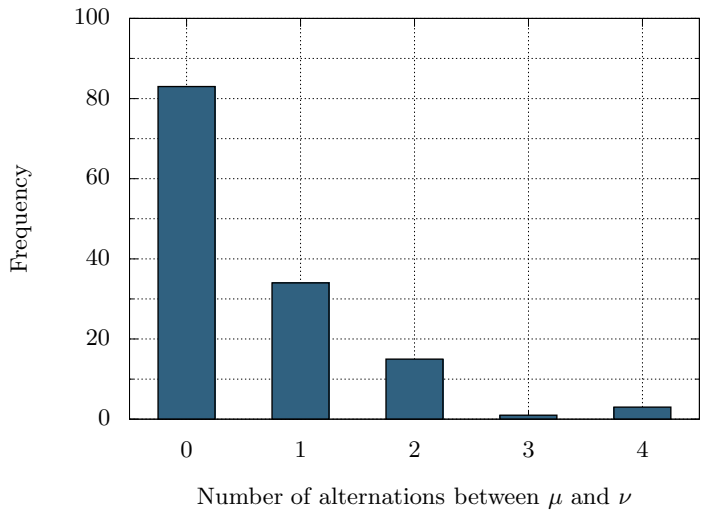}
    \caption{Distribution of the numbers of alternations}
    \label{fig:hist}
\end{figure}

\fi

\end{document}